\documentclass[10pt,journal,compsoc]{IEEEtran}

\ifCLASSOPTIONcompsoc
  \usepackage[nocompress]{cite}
\else
  \usepackage{cite}
\fi

\usepackage[cmex10]{amsmath}
\usepackage{amssymb}
\usepackage{amsthm}
\usepackage{graphicx}
\usepackage{array,color}

\usepackage{amsmath}
\usepackage{stfloats}
\usepackage{graphicx}
\usepackage{subfigure}
\usepackage{tabularx}
\usepackage{epsfig,epsf,color,balance,cite}
\usepackage{algorithmic}
\usepackage{algorithm}
\usepackage{bm}
\usepackage{textcomp}
\usepackage{color}
\usepackage{multirow}
\usepackage{cite}
\usepackage{enumerate}
\usepackage{cases}
\usepackage{color}
\usepackage{url}

\begin{document}

\title{A Caching Strategy Towards Maximal D2D Assisted Offloading Gain}

\author{Yijin Pan, Cunhua Pan, Zhaohui Yang, Ming Chen,
		and Jiangzhou Wang,~\IEEEmembership{Fellow,~IEEE}
\IEEEcompsocitemizethanks{\IEEEcompsocthanksitem Y. Pan, Z. Yang and M. Chen are with the National Mobile Communications Research Laboratory, Southeast University, Nanjing 211111, China. \protect\\
Email:\{panyijin,yangzhaohui,chenming\}@seu.edu.cn.
\IEEEcompsocthanksitem C. Pan is with the School of Electronic Engineering and Computer Science, Queen Mary, University of London, London E1 4NS, UK. \protect\\Email: c.pan@qmul.ac.uk.
\IEEEcompsocthanksitem J.Wang are with School of Engineering and Digital Arts, University of Kent, Canterbury, CT2 7NT,  United Kingdom.
\protect\\Email: J.Z.Wang@kent.ac.uk}}

\IEEEtitleabstractindextext{%
\begin{abstract}
Device-to-Device (D2D) communications incorporated with content caching have been regarded as a promising way to offload the cellular traffic data.
In this paper, the caching strategy is investigated to maximize the D2D offloading gain with the comprehensive consideration of user collaborative characteristics as well as the physical transmission conditions.
Specifically, for a given content, the number of interested users in different groups is different, and users always ask the most trustworthy user in proximity for D2D transmissions.
An analytical expression of the D2D success probability is first derived, which represents the probability that the received signal to interference ratio is no less than a given threshold.
As the formulated problem is non-convex, the optimal caching strategy for the special unbiased case is derived in a closed form, and a numerical searching algorithm is proposed to obtain the globally optimal solution for the general case.
To reduce the computational complexity, an iterative algorithm based on the asymptotic approximation of the D2D success probability is proposed to obtain the solution that satisfies the Karush-Kuhn-Tucker conditions.
The simulation results verify the effectiveness of the analytical results and show that the proposed algorithm outperforms the existing schemes in terms of offloading gain.
\end{abstract}

}
\maketitle

\IEEEpeerreviewmaketitle
\IEEEraisesectionheading{\section{Introduction}\label{sec:introduction}}

\IEEEPARstart{T}he current wireless network infrastructure is expected to provide explosively increasing data traffic in the near future \cite{index2015global}.
Although some crucial technologies, such as multiple input multiple output (MIMO) \cite{6180090,5770661}, are promising to increase spectral efficiency, ``offloading'' is another important technology to increase network capacity.
Device-to-Device (D2D) communications have shown great potential to offload traffic from the network backbone\cite{rebecchi2015data,chen2015energy}.
Through D2D communications, mobile devices can directly connect with other devices in proximity.
When a mobile device requests a content, if the nearby devices happened to have the content in local storages, the content can be delivered via the D2D communications without incurring the cost of cellular bandwidth.
Moreover, by reducing the distance between contents and requesters, the download delay and network traffic load can be effectively reduced.

The offloading benefits are achievable only when the nearby users happened to have the requested content in their local storage.
Thus, in order to reap the offloading gain, the content caching strategy should be carefully designed and investigated.
For instance, the D2D caching strategy in \cite{JiMJasc} was designed to satisfy both outage probability requirement and per-user throughput requirement.
An optimal D2D collaboration distance to maximize the D2D link number was investigated in \cite{NeginTWC}.
However, these results were obtained based on the simplistic grid network model.
An emerging alternative is to model the locations of users and base stations (BSs) as the Poisson point process (PPP), which gives a tractable expression of the received signal-to-interference-ratio (SIR).
By adopting the PPP model, a probabilistic caching strategy was optimized in \cite{Malak} to maximize the successful receptions of contents.
An optimal caching and transmit time scheduling to maximize successful offloading probability was investigated in \cite{scheduling}.
These approaches mainly focus on the D2D assisted content delivery process, where the caching strategy is optimized with the constraints of SIR requirements, delay, and storage size.

 As D2D communications heavily rely on the collaborative interactions of mobile users, an increasing number of contributions start to exploit the human characteristics and social information of the mobile device holders to enhance the network performance.
It has been widely recognized that the content preferences of the different individuals have a significant influence on D2D assisted content delivery process.
The content preference reflects the different interests of human users for the same content \cite{li2014multiple}.
For instance, in \cite{Cooperative2017}, users were categorized into different groups according to the content request distributions, and the caching strategy was designed for each user group.
The homogeneous marked PPP distribution was employed in \cite{mustafa2017spatial} to model the spatial and social relations of D2D devices, where the Zipf based thinning is applied to obtain the distributions of users with different preferences.
When the user preferences are unknown, a learning approach was proposed in \cite{chen2017caching} to estimate the preference for the design of caching strategy.

Except for the preference, another big challenge is understanding the impact of trustworthiness information on the user cooperation behavior.
In the practical implementation, there may be malicious users that can corrupt and manipulate the content, however, which are expected to forward to other users.
If these untrustworthy factors are ignored, the manipulation and fraud of the contents are very likely to outweigh the benefits of user cooperations.
Actually, the D2D users can be aware of the trustworthiness information of their D2D partner from various ways \cite{li2012routing,jaho2013social}.
Moreover, the BS can act as the centralized trust authority since the data corruption and disobedience can be identified by checking the timeout\cite{militano2017nb,jedari2018survey,7239640}.
Consequently, to cope with the threats from malicious users and block the insecure links, the rational users will always ask the most trustworthy nodes in proximity for the reliable content transmission \cite{jedari2018survey}.

Furthermore, apart from the impacts of user preference and trustworthiness, the successful D2D transmissions also depend on the physical communication conditions.
Unfortunately, existing contributions mainly considered social characteristics in D2D caching, while ignoring the physical transmission conditions such as path loss, channel fading and underlaid cellular interference\cite{cao2016social, SocialChen, Hypergraph, SociallyWang}.
For instance, the authors in \cite{7875158} investigated content hit probability, which only considered the user distance, and the cache management scheme in \cite{8403950} was only based on the content popularity.
Recently, there have been some approaches starting to jointly explore the impacts of the physical communication conditions and the social characteristics on the D2D assisted offloading performance, such as \cite{yi2018incentive,8267090}.
In \cite{yi2018incentive}, the potential of each UE in D2D offloading was modeled by a empirical-based function, which decreases with the power cost and increases with the users' social influence.
Although this simplified model of user's social characteristics can make the formulated problem tractable, the impacts of the aforementioned preference and trustworthiness constraints on the offloading performance is still unclear.
The social relationship among users was assumed to be a decreasing function of their physical distance in \cite{8267090}, however, this model is not applicable to the scenarios where strangers are geographically adjacent, such as gyms and shopping malls.
Therefore, how to design efficient user caching strategy with the presence of the different user interests and their trustworthiness as well as the physical transmission constraints is still an unsolved problem.

In this paper, the user caching strategy to maximize the offloading gain is investigated by considering the user preference, different levels of trustworthiness, and the physical transmission conditions.
In the considered model, the number of interested users for a given reference content is different for different user groups, and users are rational so that they always incline to ask the most trustworthy user in proximity for D2D transmission.
To conduct a successful D2D transmission, the received SIR from the users who have cached contents should be no less than a given threshold.
Then, the offloading gain is evaluated by the average successful receptions of contents transmitted via D2D links.
With the target to maximize the offloading gain, the caching density in each user group is then optimized by solving the formulated non-convex optimization problem.

The main contributions are summarized as follows:
\begin{itemize}
	\item
	 The D2D success probability is analytically derived, which is expressed as a function of the interested user density, the trustworthiness of users and physical transmission conditions. In the work, users with different preferences can be geographically adjacent, and the user providing D2D offloading is selected according to the biased received power, so that only the user with the sufficient short transmission distance and the sufficient large trustworthiness can be selected.
	
	\item
	 With the target to maximize the offloading gain, the optimal caching strategy is first developed for a special case where the trust bias of each group is the same. Then, a numerical searching algorithm based on the gradient projection and two-dimensional searching is proposed to obtain the globally optimal solution for the general case.
	
	\item
	 To reduce the computational complexity, an asymptotic approximation of the D2D success probability is proposed by relaxing the D2D range restriction. Based on the obtained asymptotic approximation, an iterative algorithm is proposed to obtain the solution that satisfies the Karush-Kuhn-Tucker (KKT) conditions.
	
	\item
	The simulation results validate the effectiveness of the derived D2D successful transmission probability and the corresponding asymptotic approximation.
	Then, it is shown that our proposed caching scheme achieves a larger offloading gain than the existing caching strategies.
\end{itemize}	

The rest of this paper is organized as follows. The system model is described in Section II, and the derived performance metric and optimization problem is presented in Section III. The optimization of caching strategy is addressed in Section IV. Finally, we present the simulation results in Section V, and conclude our work in Section VI.

Notations: Capital and lower-case bold letters denote matrices and vectors, respectively.
The superscripts $[X]^{-1}$ and $[X]^T$ denote inversion and transpose, respectively.
$\bm{I}_M$ represents an $M\times M$ identity matrix.
$\bm{1}_M$ denotes an $M \times 1$ vector of ones.
`s.t.' is short for `subject to'.
$\mathbb{E}_{\{x\}}\{y\}$ means the expectation of y over x.

\section{System Model}

Consider a network where the mobile users are equipped with caches to store their interested contents.
The cache-enabled users can share the cached content with others via D2D communications.
When BS needs to deliver a content (here referred as the reference content
\footnote{The proposed scheme focuses on the scenario of delivering a certain piece of content, and mobile users' caching capacities are sufficient for the reference content.}) to the interested users, cellular traffic can be effectively offloaded by D2D links.
Specifically, a subset of users store the reference content in their local cache, and then the other interested users can get this content via D2D communications, instead of downloading from the BS.

To identify the impact of social characteristics, the users are classified into groups based on the user profile and social similarity.
The user profile contains the key features that can reflect the users' preferences and trustworthiness for cooperation, such as the personal data (name, age ...), interests and preferences (searching keywords, language, habits ...), system logs (browsing history, historical success cooperations... ).
Based on the user profile, the emerging convolutional neural network (CNN) \cite{krizhevsky2012imagenet} is employed for user classification.
First, some typical users belonging to the particular types are selected as the training data set.
According to the social characteristics, assume that users can be classified into $M$ disjoint groups.
The set of group indices is denoted by $\mathcal{M}=\{1,2, \cdots ,M\}$.
Then, the coefficients of CNN is trained and optimized by employing the profile data of these typical users, which imply the implicit features to distinguish user.
With the trained CNN, the left users are filtered and classified into the corresponding groups.

It is worth pointing out that the classification of users is based on the behavior similarity rather than the proximity of geographical location information.
The reason is that the users with different preferences are not strictly separated geographically. Instead, users of different social types can be adjacent, e.g., in a stadium or large shopping mall.
Consequently, for the reference content, the locations of interested users in group $m$ are deployed according to the PPP distribution denoted by ${\Phi}_m$ with density $\lambda_m$.
In other words, the interested user density $\lambda_m$ also reflects the user preference of different groups.
Under this assumption, it is possible for a user to have users from other groups in proximity.

We use ``UTs'' to denote the users who caches the reference content.
According to the thinning property of PPP\cite{chiu2013stochastic}, the locations of UTs in group $m$ also follow the PPP distribution, and the corresponding density is denoted by $c_m$, where $c_m \in [0,\lambda_m]$.
It is worth pointing out that the caching density can be interpreted as the product of the user density $\lambda_m$ and a caching probability $q_m$, i.e., $c_m = \lambda_m q_m$, where the caching probability $q_m$ represents the probability that each user in the group $m$ caches the given reference content.
To practically implement a given caching strategy $c_m$, each user in the group $m$ generates a uniformly distributed random number within the range of $[0,1]$.
Then, if the generated random number is less than or equal to the caching probability $q_m$, this user should cache the content.
Otherwise, i.e., the generated random number is larger than the caching probability $q_m$, this user does not cache the content.
In this way, a thinned user group with density $c_m$ is readily obtained.

In addition, there are some users who did not cache the reference content, but they are also interested in it.
Consequently, they will request the reference content from the nearby UTs, and this kind of users are named as ``URs''.
The density of URs in group $m$ is $\lambda_m-c_m$.
According to the Slivnyak's theorem \cite{chiu2013stochastic}, the properties observed by a typical point of the PPP, $\phi$, is the same as those observed by the point at origin in the process $\phi \cup \{0\}$.
As a result, we consider a typical UR located in the origin as a reference UR.

\begin{figure}
	\centering
	\includegraphics[width=1\linewidth]{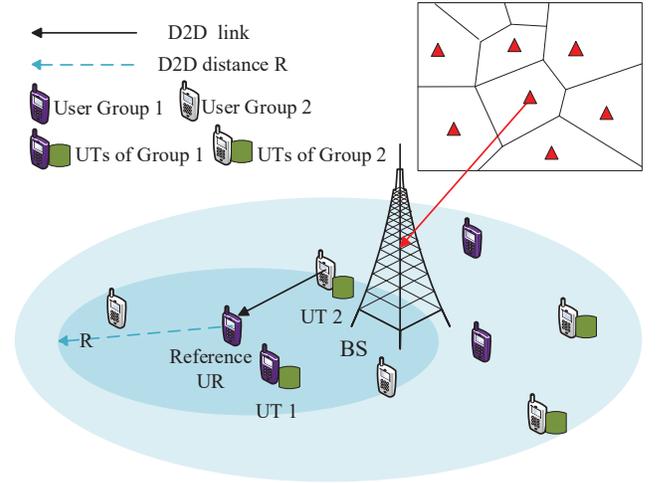}
	\vspace{-0.5em}
	\caption{ D2D communications between users from different groups in a typical cell. The upper plot shows the BSs (red triangles) in the adjacent cells, and the locations of BS locations follow the PPP distribution.}
	\label{fig0}
	\vspace{-1.5em}
\end{figure}

{As the BS can be a centralized authority of the levels of users' trustworthiness, D2D users can be aware of the trustworthiness information of their D2D partner.
To cope with the threats from malicious users, URs are inclined to ask the most trustworthy UTs in proximity for content transmission.
Consequently, the serving UT may not be the nearest UT.}
As shown in Fig. \ref{fig0}, for the reference UR, there are two UTs in proximity, i.e., UT 1 and UT 2 are from group 1 and group 2, respectively.
Due to the fact that group 2 is more trustworthy than group 1, the UR will request the content from UT 2, though UT 1 is closer than UT 2.

To reflect the URs' behavior of selecting the serving UT based on trustworthiness, the trust bias is introduced to show how trustworthy the users in group $m$ can be.
Let $B_m$ denote the trust bias value of group $m$, and the	value of $B_m$ is in the range of $[0, 1]$ \cite{chen2014dynamic,militano2016trust}.
A larger $B_{m}$ corresponds to higher trust for the UTs in group $m$, where ``$1$'' indicates complete trust and ``$0$'' means totally distrust.
To evaluate the trust bias $B_m$ of group $m$, we consider that the trustworthiness depends mainly on the historical success transmissions \cite{militano2016trust}.
Within a detection time window, the BS counts the average number of the correctly verified contents sent by each user in group $m$, and denote it as $\mathbb{N}_m$.
Then, the trust bias can be obtained by normalizing the $\mathbb{N}_m$ for all the $M$ groups, i.e.
\begin{equation}
B_m = \frac{\mathbb{N}_m}{\sum_{m=1}^{M}\mathbb{N}_m}.
\end{equation}
Moreover, we also have $\sum_{m=1}^{M}B_m =1$.

Meanwhile, the successful D2D cooperation not only depends on the user social characteristics but also depends on the physical communication conditions.
To jointly incorporate this trust-driven user behavior and the transmission distance constraints together, we define the biased-received-power (BRP).
Specifically, the BRP for the UTs in group $m$ is defined as
\begin{equation}
P_m =  p_t B_m d_{m}^{-\alpha},
\end{equation}
where $p_t$ is the D2D transmit power,	$d_{m}$ is the distance from the nearest UT in group $m$ to the reference UR, $\alpha$ is the path-loss exponent.

Then, the D2D links are selected according to averaged maximum biased-received-power (BRP), i.e., UTs with maximum BRP will serve the D2D transmission.
According to the maximum BRP criteria, faraway UTs or the UTs with low trustworthiness will not be selected by the UR for requesting content.
In this way, the physical transmission constraints and the trustworthy constraints are jointly considered.
Obviously, when $B_m =B_k, \forall m,k \in \mathcal{M}$, it reduces to the no biasing case, and the UR is served by the UT with strongest average received power. \footnote{Selecting the nearest UR to be the serving UR may not be the optimal but is commonly adopted in current works due to its tractability\cite{7605533,deng2018benefits}.}

Similar to \cite{7893755,7470622,8057276}, we assume that the D2D communications are underlaid with cellular downlinks, where the locations of BSs also follow a PPP distribution denoted by ${\Phi}_B$ with density $\lambda_B$ and $\lambda_B \ll \lambda_m$.
To reduce the system interference, we assume that UTs in each group reuse the same frequency, while the UTs in different groups adopt orthogonal frequencies\cite{5281762,6094142,6180097}.
The assumption is justified when all UTs in each group are uniformly generated in the network.
By allocating UTs from different groups with orthogonal frequencies, the distances between UTs with the same channel can be efficiently enlarged, and the received interference in URs can be notably reduced.
We consider the interference limited scenario where the thermal noise can be ignored.
When the serving UT is from group $m$, the received SIR of the reference UR is
\begin{equation}\label{SIR}
\text{SIR}_{m} = \frac{p_t |h_{m}|^2 d_{m}^{-\alpha}}
{ \sum_{i \in {\Phi}'_m\setminus\{T^m_0\}} p_t |h_{i}|^2 d_{i}^{-\alpha} +  \sum_{j \in {\Phi}_B} p_B |h_{j}|^2 d_{j}^{-\alpha}},
\end{equation}
where $h_{m}$ represents the channel fading from serving UT to UR, and the channel gain $|h_{m}|^2$ is exponentially distributed with united mean value, i.e., $|h_{m}|^2 \sim \exp(1)$.
Term $\sum_{i \in {\Phi}'_m\setminus\{T^m_0\}} p_t |h_{i}|^2 d_{i}^{-\alpha}$ represents the interference from other active UTs in group $m$ except the serving UT denoted by $T^m_0$, where ``active UTs'' refers to the UTs that are paired up with URs conducting D2D transmissions. ${\Phi}'_m$ denotes the set of the interfering UTs in group $m$ and ${\Phi}'_m \subseteq {\Phi}_m$.
The distance from the interfering UT $i$ to UR is denoted by $d_{i}$.
Term $\sum_{j \in {\Phi}_B} p_B |h_{j}|^2 d_{j}^{-\alpha}$ represents the interference from underlaid BSs.
The distance from the interfering BS $j$ to UR is denoted by $d_{j}$.
$|h_{i}|^2$ and $|h_{j}|^2$ are exponentially distributed channel gains with unit mean.

\newtheorem{proposition}{\textbf{Proposition}}
\newtheorem{theorem}{\textbf{Theorem}}
\newtheorem{lemma}{\textbf{Lemma}}
\newtheorem{corollary}{\textbf{Corollary}}

\section{Performance Metrics and Problem Formulation}

The URs in each group are those who do not have the reference content in their cache so that the density of URs in group $m$ is $\lambda_m -c_m$.
For a UR from one of the $M$ groups, its serving UT can come from any of the $M$ groups.
Since different types of users can be geographically adjacent, the D2D success probability for any URs should depend on the cache densities in all the $M$ groups.
Let $\mathbb{P}_{m}^{th}[\text{SIR}_{m}>\gamma_{th}]$ denoted the D2D success probability in the $m$-th group, where $\gamma_{th}$ is the required received SIR threshold.

Then, based on the law of total probability, the D2D success probability of a UR is calculated as
\begin{equation}
\mathbb{P}_{s} = \sum_{m=1}^M \mathcal{P}_m \mathbb{P}_{m}^{th}[\text{SIR}_{m}>\gamma_{th}],\label{Eq_ps}
\end{equation}
where $\mathcal{P}_m$ represents the probability that the serving UT is from group $m$.

To measure the offloading performance for the D2D assisted offloading network, we are interested in the average number of contents successfully transmitted via D2D links, namely D2D-aided offloading gain.
{
Specifically, the offloading gain of the $m$-th group is
\begin{equation}
\mathcal{U}_m = \mathbb{P}_{s} (\lambda_m-c_m),
\end{equation}
and the total offloading gain is
\begin{equation}
\mathcal{U} = \sum_{i=1}^M(\lambda_i-c_i) \mathbb{P}_{s}.
\end{equation}}

To obtain the expression of $\mathcal{U}$, we derive the expressions for $\mathcal{P}_m$ and $\mathbb{P}_{m}^{th}$ in the following propositions.

\begin{proposition} \label{pro_1}
	The probability that the serving UT is from group $m$ is
	\begin{equation}
	\mathcal{P}_m = \frac{{B_m}^{\frac{2}{\alpha}}c_m}{\sum_{i=1}^M{{B_i}^{\frac{2}{\alpha}} c_i}}
	\left(1- \exp \left(-\pi {B_m}^{-\frac{2}{\alpha}} \sum_{i =1}^M c_i B_i^{\frac{2}{\alpha}} R^2\right)\right), \label{pm1}
	\end{equation}
	where $R$ is the maximum D2D transmission distance.
\end{proposition}
\begin{proof}
	According to the definition of BRP, we have
	\begin{eqnarray}
	\mathcal{P}_m &&\!\!\!\!\!\!\!\!= \mathbb E_{d_m} \left[P_m(d_m) > \max _{i \neq m} P_i (d_i)\right]  \\
	&&\!\!\!\!\!\!\!\!=\int _{0} ^{R} \prod _{i=1, i\neq m}^M \mathbb P\left[d_i > \left(\frac{B_i}{B_m}\right)^{\frac{1}{\alpha}} r \right] f_{d_m} (r) \mathrm d r \label{Eq_Pro1},
	\end{eqnarray}
	where $R$ is the maximum D2D transmission distance.
	According to the null probability of the 2D Poisson process \cite{chiu2013stochastic},
	\begin{eqnarray}
	\mathbb P\left[d_i > \left(\frac{B_i}{B_m}\right)^{\frac{1}{\alpha}} r \right]
	= \exp \left(-\pi c_i  \left(\frac{B_i}{B_m}\right)^{\frac{2}{\alpha}} r^2\right). \label{cdf}
	\end{eqnarray}
	
	Similarly, the probability density function (PDF) of $d_m$ denoted by $f_{d_m} (r)$ is
	\begin{equation} \label{pdf}
	f_{d_m} (r)= \frac{\mathrm d (1-\mathbb{P}[d_m > r])}{ \mathrm d r }= 2 \pi c_m r \exp(-\pi c_mr^2).
	\end{equation}
	Substituting (\ref{cdf}) and (\ref{pdf}) into (\ref{Eq_Pro1}), we have
	\begin{eqnarray}
	\mathcal{P}_m &&\!\!\!\!\!\!\!\!\!\!\!\!\!=\int _{0} ^{R}  2 \pi c_m r \exp \left(-\pi \sum_{i =1}^M c_i \left(\frac{B_i}{B_m}\right)^{\frac{2}{\alpha}} r^2 \right) \mathrm d r \\
	&&\!\!\!\!\!\!\!\!\!\!\!\!\!= \frac{{B_m}^{\frac{2}{\alpha}}c_m}{\sum_{i=1}^M{{B_i}^{\frac{2}{\alpha}} c_i}}
	\left(\!\!1-\!\!\exp \left(\!\!-\pi {B_m}^{-\frac{2}{\alpha}}\!\!\sum_{i =1}^M c_i B_i^{\frac{2}{\alpha}} R^2\!\!\right)\!\!\right).
	\end{eqnarray}
\end{proof}

{According to $\mathcal{P}_m$ given in Proposition 1, we have $	\lim_{B_m \to 0 } \mathcal{P}_m  =  0, \lim_{d_m \to \infty} \mathcal{P}_m  =  0$.
In other words, the faraway UTs or the UTs with low trustworthiness will not be selected by UR for requiring D2D transmission.
Therefore, to obtain the expression of $\mathbb{P}_{m}^{th}$, we first need to characterize the ratio of active UTs in each group.
Let ${\Phi}'_i, \forall i \in \mathcal{M}$ denote the set of active users in group $i$ conducting D2D transmissions, and $\rho_i$ represents the ratio of active users in group $i$.}
Since only the UT will conduct a D2D link, the distribution of ${\Phi}'_i, \forall i \in \mathcal{M}$ is a PPP with density $ \rho_i c_i$.
Note that $\rho_i$ is also a function of $c_m$ for each group, and it can be approximated in the following proposition.

\begin{proposition} \label{pro_21}
	The ratio of active users in group $m$ is
	\begin{multline}
	\rho_m = 1-\left(1+\frac{\mathcal{P}_m\sum_{m=1}^M (\lambda_m-c_m)}{3.5c_m} \right)^{-3.5} \\ \frac{\Gamma(3.5,(\sum_{m=1}^M(\lambda_m -c_m)+\frac{3.5c_m}{\mathcal{P}_m})\pi R^2)}{\Gamma(3.5, \frac{3.5c_m}{\mathcal{P}_m}\pi R^2)},\label{Eq_Rho}
	\end{multline}
	where $\Gamma(a,b) = \int_0^{b}t^{a-1} e^{-t} \mathrm{d}t$.
\end{proposition}
\begin{proof}
	A UT will be active when it is associated with at least one UR via D2D links.
	According the definition of being active, it will be easy to calculate $\rho_i$ via its complementary event.
	We denote the $\bar \rho_m$ as the probability that a UT in group $m$ is not active, i.e., the UT does not transmit to any UR.
	Similar as \cite{lee2012coverage,singh2013offloading,liu2017caching}, the coverage region of a UT is modelled as the Poisson Voronoi cell.
	According to \cite{ferenc2007size}, the probability density function (PDF) of a typical Voronoi cell area $S$ in a Poisson random tessellation is given by
	\begin{equation}
	f_S(s) = \frac{3.5^{3.5}}{\Gamma(3.5,\infty)}\hat{S}^{-3.5} s^{2.5} \exp\left(-3.5 \frac{s}{\hat{S}}\right),
	\end{equation}
	where $\Gamma(a,b) = \int_0^{b}t^{a-1} e^{-t} \mathrm{d}t$ and $\hat{S}$ is the mean value of random variable $S$.
	According to the Remark 2 in \cite{singh2013offloading}, the probability $\mathcal{P}_m$ can be viewed as the average fraction of the total area covered by the association regions of UTs in group $m$ by using the ergodicity of the PPP.
	We assume the total area is $S_{tot}$, and the average total number of UTs is $c_mS_{tot}$.
	Similar to \cite{liu2017caching}, the mean association region of a typical UT in group $m$ is approximated by
	\begin{equation}
	\hat{S}_m =\frac{\mathcal{P}_mS_{tot}}{S_{tot}c_m} = \frac{\mathcal{P}_m}{c_m}.
	\end{equation}
	Then, $\tilde{\rho}_m$ is calculated by
	\begin{eqnarray}
	\tilde{\rho}_m &&\!\!\!\!\!\!\!\! = \mathbb{P}(\text{no UR within }S| S \leq \pi R^2) \\
	&&\!\!\!\!\!\!\!\! = \frac{\int_0^{\pi R^2} \exp(- s \sum_{m=1}^M(\lambda_m -c_m) ) f_S(s) \mathrm d x}{\int_0^{\pi R^2} f_S(s) \mathrm d s}\\
	&&\!\!\!\!\!\!\!\! =  \left(1+\frac{\mathcal{P}_m\sum_{m=1}^M (\lambda_m -c_m)}{3.5c_m} \right)^{-3.5} \nonumber \\
	&&\frac{\Gamma(3.5,(\sum_{m=1}^M(\lambda_m -c_m)+\frac{3.5c_m}{\mathcal{P}_m})\pi R^2)}{\Gamma(3.5, \frac{3.5c_m}{\mathcal{P}_m}\pi R^2)}.
	\end{eqnarray}
	Consequently, the ratio of active users in group $m$ is obtained as
	\begin{equation}
	\rho_m = 1- \tilde{\rho}_m.
	\end{equation}
\end{proof}

Given the ratio of active users in each group, the user density of ${\Phi}'_m$ denoted by $\rho_m c_m$ can be obtained accordingly.

\begin{proposition}\label{pro_2}
	If the serving  ``UT'' is from group $m$, the probability that the received SIR is beyond the threshold is
	\begin{equation}
	\mathbb{P}_{m}^{th}[\text{SIR}_{m}>\gamma_{th}] =\frac{\pi R^2 c_m}{\mathcal{P}_m\varphi_m} (1-\exp(-\varphi_m)),  \label{pm2}
	\end{equation}
	where $$\varphi_m = \pi R^2 \left(\sum_{i=1}^{M} c_i \left(\frac{B_i}{B_m}\right)^{\!\!\!\frac{2}{\alpha}} + \lambda_B \theta_B +c_m \rho_m \theta_I\right), $$
	$$\theta_I= {\gamma_{th}}^{\frac{2}{\alpha}} \int_{\gamma_{th}^{-\frac{2}{\alpha}}}^{\infty}\frac{1}{1+ u^{\frac{\alpha}{2}}} \mathrm d u , \theta_B= {\left(\gamma_{th} \frac{p_B}{p_t}\right) }^{\frac{2}{\alpha}} \int_{0}^{\infty}\frac{1}{1+ u^{\frac{\alpha}{2}}} \mathrm d u .$$
\end{proposition}
\begin{proof}
	For simplicity, we define $I_{m} = \sum_{i \in {\Phi}'_m\setminus\{T^m_0\}} |h_{i}|^2 d_{i}^{-\alpha}$ and $Q =  \sum_{j \in {\Phi}_B} |h_{j}|^2 d_{j}^{-\alpha}$.
	To avoid confusion with Appendix A, when the serving UT is from group $m$, we denote $R_m$ as the distance between the serving UT and the reference UR.
	Then, we have
	\begin{eqnarray}
	&&\!\!\!\!\!\!\!\!\mathbb{P}_{m}^{th}[\text{SIR}_{m}>\gamma_{th}] \nonumber \\
	&&\!\!\!\!\!\!\!\!=\mathbb{P}_{m}^{th}\left[ |h_{m}|^2>
	\gamma_{th} R_m^{\alpha}  \left(I_{m} + \frac{p_B}{p_t} Q\right) \right] \nonumber\\
	&&\!\!\!\!\!\!\!\!\overset{(a)}{=} \int_{0}^{R} \mathbb{E}_{I_{m}} \left[\exp \left(-\gamma_{th} x^{\alpha}I_{m}\right) \right] \nonumber\\
	&&\qquad \mathbb{E}_{Q} \left[\exp \left(-\frac{p_B}{p_t} \gamma_{th} x^{\alpha}Q\right) \right] f_{R_m}(x) \mathrm d x \nonumber\\
	&&\!\!\!\!\!\!\!\!{=}\int_{0}^{R} \mathcal{L}_{I_m} (\gamma_{th} x^{\alpha}) \mathcal{L}_{Q} \left(\frac{p_B}{p_t}\gamma_{th} x^{\alpha}\right) f_{R_m}(x) \mathrm d x.  \label{Eq_Pth}
	\end{eqnarray}
	Step (a) follows that $|h_{m}|^2 \sim \exp(1)$, and $f_{R_m}(x)$ is the PDF of distance $R_m$.
	$\mathcal{L}_{I_{m}} (s) $ is the Laplace transformation of random variable $I_{m}$ evaluated at $s$ conditioned on the distance $R_m =x$.
	According to \cite{Andrews2011}, we have
	\begin{eqnarray}
	&&\!\!\!\!\!\!\!\!\mathcal{L}_{I_{m}} (\gamma_{th} x^{\alpha}) \nonumber\\
	&&\!\!\!\!\!\!\!\!= \mathbb{E}_{{\Phi}'_i} \left[ \prod_{i \in {\Phi}'_m} \mathbb{E}_{h_{i}} \left[\exp \left(-\gamma_{th} x^{\alpha} |h_{i}|^2 d_{i}^{-\alpha} \right) \right]\right]\nonumber\\
	&&\!\!\!\!\!\!\!\!\overset{(b)}{=}\exp\left( - 2\pi c_m\rho_m \int_{x}^{\infty} (1- \right.\\
	&&\left. \mathbb{E}_{h_{i}}[\exp(-\gamma_{th}x^{\alpha} |h_{i}|^2 v^{-\alpha})]) v \mathrm d v\right) \nonumber\\
	&&\!\!\!\!\!\!\!\! \overset{(c)}{=} \exp\left(- 2 \pi c_m \rho_m\int_{x}^{\infty}\frac{\gamma_{th}}{\gamma_{th}+ (\frac{v}{x})^\alpha} v  \mathrm d v \right) \nonumber\\
	&&\!\!\!\!\!\!\!\!\overset{(d)}{=}\exp\left( -\pi c_m\rho_m x^2\theta_I \right). \label{Eq_Iap}
	\end{eqnarray}
	Step (b) follows from the probability generating functional (PGFL) of the PPP.
	Step (c) is due to the distribution of channel gain, i.e. $|h_{j}|^2 \sim \exp(1)$.
	Step (d) is obtained by employing the change of variable $u= \left(\frac{v}{x\gamma_{th}^{1/\alpha}} \right)^2$, and
	$ \theta_I=\gamma_{th}^{\frac{2}{\alpha}} \int_{\gamma_{th}^{-\frac{2}{\alpha}}}^{\infty}\frac{1}{1+ u^{\frac{\alpha}{2}}} \mathrm d u $.
	
	Similarly, we can obtain
	\begin{equation}
	\mathcal{L}_{Q} \left(\frac{p_B}{p_t}\gamma_{th} x^{\alpha}\right)=
	\exp\left( -\pi \lambda_B x^2\theta_B \right), \label{Lq}
	\end{equation}
	where $\theta_B= {\left(\gamma_{th} \frac{p_B}{p_t}\right) }^{\frac{2}{\alpha}} \int_{0}^{\infty}\frac{1}{1+ u^{\frac{\alpha}{2}}} \mathrm d u $.
	
	To obtain the final expression, we still need to derive the PDF of $f_{R_m}(x)$.
	We first derive the probability of $R_m>x$ as
	\begin{eqnarray}
	&&\!\!\!\!\!\!\!\!\mathbb{P}[R_m>x] \nonumber \\
	&&\!\!\!\!\!\!\!\!= \frac{1}{\mathcal{P}_m}\mathbb{P}[R_m>x,P_m(R_m) > \max _{i \neq m} P_i (d_i) ]  \nonumber\\
	&&\!\!\!\!\!\!\!\! =\frac{1}{\mathcal{P}_m}\int_x^{\infty} \prod_{i\neq m}^M\mathbb{P}\left[ d_i > \left(\frac{B_i}{B_m}\right)^{\frac{1}{\alpha}}r\right] f_{d_m}(r) \mathrm{d}r \nonumber	\\
	&&\!\!\!\!\!\!\!\!\overset{(f)}{=}\frac{1}{\mathcal{P}_m}\!\!\int_x^{\infty}2\pi c_m r \exp\left(-\pi\sum_{i=1}^{M} c_i\left(\frac{B_i}{B_m}\!\!\right)^{\frac{2}{\alpha}}\!\!r^2\!\!\right)\mathrm{d}r,	
	\end{eqnarray}
	where $f_{d_m}(r)$ is given in (\ref{pdf}), and step (f) follows from (\ref{cdf}).
	Then, the pdf $f_{R_m}(x)$ is
	\begin{eqnarray}
	f_{R_m}(x)&&\!\!\!\!\!\!\!\!
	= \frac{\partial (1-\mathbb{P}[R_m>x]) }{\partial x} \nonumber \\
	&&\!\!\!\!\!\!\!\! = \frac{1}{\mathcal{P}_m}2\pi c_m x \exp\left(-\pi \sum_{i=1}^{M} c_i\left(\frac{B_i}{B_m}\right)^{\frac{2}{\alpha}} x^2\right) \label{Eq_RmPdf}.
	\end{eqnarray}
	
	Applying (\ref{Eq_Iap}), (\ref{Lq}) and (\ref{Eq_RmPdf}) into (\ref{Eq_Pth}) yields
	\begin{align}
	&\mathbb{P}_{m}^{th}[\text{SIR}_{m}>\gamma_{th}] \nonumber \\
	&=\frac{1}{\mathcal{P}_m} \!\!\int_{0}^{R} \!\!\!2\pi c_m x
	\exp\left(\!\!\!-\pi (\lambda_B \theta_B +\rho_mc_m \theta_I) x^2 \right. \nonumber \\
	&\left. \qquad\quad -\pi\sum_{i=1}^{M} c_i  \left(\frac{B_i}{B_m}\right)^{\!\!\!\frac{2}{\alpha}}\!\!\!x^2\right) \!\!\mathrm d x  \nonumber \\
	&=\frac{c_m}{\mathcal{P}_m\varphi_m} (1-\exp(-\pi R^2 \varphi_m)),
	\end{align}
	where $\varphi_m = \sum_{i=1}^{M} c_i \left(\frac{B_i}{B_m}\right)^{\!\!\!\frac{2}{\alpha}} + \lambda_B \theta_B +c_m\rho_m \theta_I$.
\end{proof}

Substituting (\ref{pm1}) and (\ref{pm2}) into (\ref{Eq_ps}), we have
\begin{equation}
\mathbb{P}_{s} = \pi R^2 \sum_{m=1}^M c_mf(\varphi_m), \label{Eq_Pss}
\end{equation}
where
\begin{equation}
f(t) = \frac{1-e^{-t}}{t}. \label{ft}
\end{equation}

Based on the above results, we now formulate the following optimization problem to maximize the D2D-aided offloading gain by optimizing the caching strategy
$\bm{c}= [c_1,c_2,\cdots,c_M]^T$.
\begin{subequations}\label{P1}
	\begin{align}
	\mathcal{P}1:\mathop{\max }_{\bm{c}}
	& \quad \mathcal{U} = \left(\lambda_0-\sum_{i=1}^Mc_i\right) \mathbb{P}_{s}  \label{Obj1}\\
	\textrm{s.t.} &\quad 0\leq c_m \leq \lambda_m , \textrm{ for all } m \in  \mathcal{M}.  \label{st1}
	\end{align}
\end{subequations}
where $\lambda_0 = \sum_{i=1}^M \lambda_i $, and (\ref{st1}) are the probability constraints.

\section{Optimization of Caching Strategy}

In this section, we first consider the unbiased case, where the optimal caching strategy and offloading gain is obtained.
However, Problem $\mathcal{P}1$ is not convex in general due to the sophisticated expression of D2D success probability $\mathbb{P}_{m}^{th}$ given in Proposition 2.
For the general case, we first propose a numerical searching algorithm to obtain the globally optimal solution.
To deal with the high computational complexity, we adopt an asymptotic approximation of $\mathbb{P}_{m}^{th}$ for the large D2D distance region.
Then, an iterative algorithm with low complexity is proposed to obtain the asymptotic caching strategy.

\subsection{Unbiased Case}

In the unbiased case, there is no trust bias for association, i.e., $B_m =B_k, \forall m,k \in \mathcal{M}$.
In this case, we define $x= \sum_{i=1}^M c_i$.
Then Problem $\mathcal{P}1$ is transformed into
\begin{subequations}\label{P2}
	\begin{align}
	\mathcal{P}2:\mathop{\max }_{x,\bm{c}}
	& \quad \mathcal{U}(x,\bm{c})= \left(\lambda_0-x\right)\mathbb{P}_{s}(\bm{c},x)  \\
	\textrm{s.t.} &\quad x= \sum_{i=1}^M c_i ,  \label{st2_2} \\
	&\quad 0 \leq x \leq \lambda_0,  \label{st2_3} \\
	&\quad 0 \leq c_m \leq \lambda_m , \textrm{ for all } m \in  \mathcal{M}.
	\end{align}
\end{subequations}
Problem $\mathcal{P}2$ is still non-convex, however, we can obtain the optimal solution for this case by the following procedures.
First, Problem $\mathcal{P}2$ is decomposed into the following two subproblems.
For given $x=\bar{x}$, we can find the optimal $c_m$ for each group by solving the following subproblem.
\begin{subequations}\label{P2.1}
	\begin{align}
	\mathcal{P}2.1:\mathop{\max }_{\bm{c}}
	& \quad \mathcal{F}(\bm{c}|x=\bar{x})= \sum_{m=1}^M c_mf(\varphi_m)  \label{p21Obj} \\
	\textrm{s.t.} &\quad \bar{x}= \sum_{i=1}^M c_i ,  \label{st21} \\
	&\quad 0 \leq c_m \leq \lambda_m , \textrm{ for all } m \in  \mathcal{M}.
	\end{align}
\end{subequations}
where function $f(x)$ is defined in $(\ref{ft})$.
Note that the active ratios of different groups are identical in this case, i.e. $\rho_m =\rho$ for all $m \in \mathcal{M}$ and
\begin{equation}
\rho= 1-\left(1+\frac{\lambda_0-\bar{x}}{3.5\bar{x}}E(\bar{x})\right)^{-3.5} \frac{\Gamma\left(3.5,\left(\lambda_0- \bar{x}+\frac{3.5\bar{x}}{E(\bar{x})}\right)\pi R^2\right)}{\Gamma\left(3.5, \frac{3.5\bar{x}}{E(\bar{x})}\pi R^2\right)},
\end{equation}
where $E(\bar{x}) = 1-\exp(-\pi R^2\bar{x})$.
Consequently,  $\varphi_m$ in this case is simplified to
$\varphi_m = \pi R^2(\bar{x} + \lambda_B \theta_B+c_m \rho \theta_I)$.

Then, we need to find the optimal caching summation $x$ from the following subproblem
\begin{subequations}\label{P2.2}
	\begin{align}
	\mathcal{P}2.2:\mathop{\max }_{x}
	& \quad \left(\lambda_0-x\right) \mathcal{F}^*(\bm{c}^*|x)  \\
	\textrm{s.t.} &\quad 0 \leq x \leq \lambda_0.
	\end{align}
\end{subequations}

{
Each value of $\bar{x} \in [0,\lambda_0]$ corresponds to a maximal objective value of Problem $\mathcal{P}2.1$, which is denoted as $\mathcal{F}^*(\bm{c}^*|\bar{x})$.
If we can determine the optimal $x^*$, the optimal solution to the original Problem $\mathcal{P}2$ is readily obtained.
Then the optimal $\bar{x}$ is given by
\begin{equation}
x^* = \arg \max_{\bar{x} \in [0,\lambda_0]} \mathcal{F}^*(\bm{c}^*|\bar{x})
\end{equation}
Given the fact that Problem $\mathcal{P}2$ is nonconvex, one-dimension exhaustive search is conducted in the interval $[0,\lambda_0]$ to find $x^*$.
With a small enough searching step size, the global optimal solution $x^*$ can be obtained.}
Hence, in the rest of this section, we focus on solving Problem $\mathcal{P}2.1$.

We first prove that Problem $\mathcal{P}2.1$ is a convex optimization problem.
Since the constraints of Problem $\mathcal{P}2.1$ are linear, we only need to show that the objective function (\ref{p21Obj}) is concave.
The Hessian matrix of function $\mathcal{F}(\bm{c}|x=\bar{x})$ is denoted by $\bf{H}$.
Let $h_{m,n}$ represents the $(m, n)$th element of $\bf{H}$, and we have
\begin{equation} \label{hmn}
h_{m,n} = \left\{
\begin{array}{cl}
\theta_I\rho \pi R^2 (2f'(\varphi_m) + c_m \theta_I\rho \pi R^2 f''(\varphi_m) ) &  \text{if } m=n, \\
0 &  \text{else.}
\end{array}
\right.
\end{equation}

For function $f(t)$ defined in (\ref{ft}), we have
\begin{equation}
f'(t) =  \frac{e^{-t}(t - e^t + 1)}{t^2}, f''(t) =  \frac{e^{-t}(-2 + 2 e^t - 2 t - t^2)}{t^3}. \label{fDiv}
\end{equation}
Note that $e^t  = 1+t+ \frac{1}{2}t^2+ \cdots$ and $e^{-t}  = 1-t+ \frac{1}{2}t^2- \cdots$, so $f'(t) <0$ and $f''(t) >0$ hold.
Moreover, $c_m \theta_I\rho \pi R^2 < \varphi_m$ holds due to the definition of $\varphi _m$. Then we have
\begin{multline}
2f'(\varphi_m) + c_m \theta_I\rho \pi R^2 f''(\varphi_m)  \leq 2f'(\varphi_m) + \varphi_m f''(\varphi_m)\\
= \frac{-e^{-\varphi_m}{\varphi_m}^2}{{\varphi_m}^3} <0 \label{diff2rd}.
\end{multline}
Consequently, the Hessian matrix $\bf{H}$ is negative definite.

Since Problem $\mathcal{P}2.1$ is convex, we can obtain the globally optimal solution by solving the KKT conditions. The Lagrangian of problem $\mathcal{P}2.1$ is
\begin{equation}
\mathcal{L}_{0}(\bm{c}, \eta) =  -\sum_{m=1}^M c_m f(\varphi_m)+ \eta( \sum_{i=1}^M c_i-\bar{x}),
\end{equation}
where $\eta$ is the dual variable associated with constraint (\ref{st21}).
Then, we have the following KKT conditions as
\begin{eqnarray}
\frac{\partial \mathcal{L}}{\partial c_m} = -f(\varphi_m) - \theta_I\rho \pi R^2c_mf'(\varphi_m) + \eta=0,&\!\!\!\!\!\!\forall m \in \mathcal{M}  \\
\eta( \sum_{i=1}^M c_i-\bar{x})=0.&\label{kktp2}
\end{eqnarray}

Define function $g(c_m)$ as
\begin{equation}
g(c_m)= f(\varphi_m(c_m)) +  \theta_I \rho\pi R^2c_mf'(\varphi_m(c_m)).\label{gcm}
\end{equation}
According to (\ref{fDiv}), we have $\theta_I\rho \pi R^2c_mf'(\varphi_m) \geq \varphi_mf'(\varphi_m)$.
Consequently, we have
\begin{equation}
g(c_m) \geq f(\varphi_m(c_m)) +  \varphi_m f'(\varphi_m(c_m)) =e^{\varphi_m} >0.
\end{equation}

Then, it is inferred that $\eta >0$.
In addition, $g(c_m)$ is a monotonically decreasing function with respect to $c_m$ according to (\ref{diff2rd}).
As a result, $-g(x) +\eta =0$ has one root, which is denoted by $C_0(\eta)$.
The optimal solution is given by
\begin{equation}
c^*_m = [C_0(\eta)]_0^{\lambda_m}, \label{cmNB}
\end{equation}
where $[x]_a^b = \max\{a,\min\{x,b\}\}$.

However, we still need to obtain the value of $C_0(\eta)$.
Since $-g(x) +\eta =0$ is a transcendental equation, it is hard to directly obtain $C_0(\eta)$ as an explicit function of
$\eta$.
However, we can still determine the value of $C_0(\eta)$ by substituting (\ref{cmNB}) into (\ref{kktp2}).
Without loss of generality, the groups are sorted according to the increasing order of $\lambda_m$, i.e. $\lambda_1 \leq \lambda_2 \leq \cdots \lambda_M$.
Then, we have
\begin{equation}
C_0(\eta) =
\left\{\begin{matrix}
\frac{\bar{x}}{M} & \text{if } \frac{\bar{x}}{M} \leq \lambda_1, \\
\frac{\bar{x}-\sum_{i =1}^n \lambda_i}{M-n} &\!\!\!\text{if} \lambda_n < \frac{\bar{x}- \sum_{i =1}^n \lambda_i}{M-n} \leq \lambda_{n+1},
\end{matrix}\right. \label{copt}
\end{equation}
where $n = 1,\cdots, M-1$.

We can see that $C_0(\eta)$ is highly dependent on  $\bar{x}$ and the interested user density $\lambda_m$ in each group.
When there is no social trust bias, (\ref{copt}) shows that the reference content should be cached as evenly as possible across user groups.
With $\bm{c}^*$ for given $\bar{x}$,  the optimal $x^*$ can be found by a one-dimension search over interval $[0,\lambda_0]$.

{In this case, the computational complexity is dominated by the one-dimensional search for finding the optimal $x^*$.
Therefore, the total computational complexity of finding optimal solutions is $\mathcal{O}\left( \left \lfloor \frac{\lambda_0}{\delta_x} \right \rfloor \right)$, where $\delta_x$ is the searching stepsize and $\left \lfloor \cdot \right \rfloor$ represents rounding down. }

\subsection{General Case}

In the general case, due to the sophisticated $\mathbb{P}_{s}$ given in (\ref{Eq_Pss}), the objective function $\mathcal{U}$ of Problem $\mathcal{P}1$ is not concave.
{For the simplicity of presentation, $v_m$ is in introduced to represent ${B_m}^{\frac{2}{\alpha}}$, i.e., $v_m = {B_m}^{\frac{2}{\alpha}}$, which is the weight of $c_m$.
Problem $\mathcal{P}1$ is nonconvex in general, but we can still obtain its globally optimal solution by the following procedure. }

{
To transform the original Problem $\mathcal{P}1$ into a convex form,
we first define
$$x= \sum_{m=1}^M c_m,y = \sum_{m=1}^M v_mc_m,$$
where $x$ and $y$ represent the summation of UTs' density in each group and the weighted sum of the UTs' density, respectively.}

For given $x= \bar{x}$ and $y= \bar{y}$, Problem $\mathcal{P}1$ is simplified to
\begin{subequations}\label{P01}
	\begin{align}
	\mathcal{P}3: \mathop{\max }_{\bm{c}}
	& \quad  \mathcal{F}(\bm{c}|\bar{x},\bar{y}) = \sum_{i=1}^M c_m f(\varphi_m(c_m) ) \label{p01Obj} \\
	&\quad \sum_{m=1}^M c_m = \bar{x},  \label{p01st1} \\
	&\quad \sum_{m=1}^M v_mc_m = \bar{y}, \label{p01st2}\\
	&\quad 0 \leq c_m \leq \lambda_m , \textrm{ for all } m \in  \mathcal{M}.
	\end{align}
\end{subequations}

In this case, with given $\bar{x}$ and $\bar{y}$, the active ratio $\rho_m$ has been determined as
\begin{multline}
\rho_m = 1-\left(1+\frac{v_m(\lambda_0-\bar{x})}{3.5\bar{y}} E_m(\bar{y}) \right)^{-3.5} \\ \frac{\Gamma\left(3.5,\left(\lambda_0-\bar{x}+\frac{3.5\bar{y}}{v_mE_m(\bar{y})}\right)\pi R^2\right)}
{\Gamma\left(3.5, \frac{3.5\bar{y}}{v_mE_m(\bar{y})}\pi R^2\right)},
\end{multline}
where $E_m(\bar{y}) = 1-\exp\left(-\pi R^2\frac{\bar{y}}{v_m}\right)$.

Consequently, we have
\begin{equation}
\varphi_m(c_m) =\pi R^2 \left(\frac{\bar{y}}{v_m}+ \lambda_B \theta_B +\rho_mc_m \theta_I\right)
\end{equation}
It is easy to know that the Hessian matrix of objective function (\ref{p01Obj}) is negative definite, which is similar to (\ref{hmn}) and omitted here.
Consequently, objective function (\ref{p01Obj}) is concave and Problem $\mathcal{P}3$ is a convex problem.

The optimal solution $\bm{c^*}=[c^*_1, \cdots, c^*_M]^T$ to Problem $\mathcal{P}3$ can be found by the gradient projection method \cite{bertsekas1999nonlinear}.
Define $\bm{v} =[v_1, \cdots, v_M]^T$, and we rewrite the equality constraints (\ref{p01st1}) and (\ref{p01st2}) as
\begin{equation}
\bm{N}\bm{c} = \bm{b},
\end{equation}
where $\bm{N}$ is the $2\times M$ coefficient matrix of the equality constraints, i.e., $\bm{N}= [\bm{1}_M,\bm{v}]^T$, and $\bm{b} = [\bar{x}, \bar{y}]^T$.
Then, the gradient projection update direction of $\bm{c}$ in $t$-th step is
\begin{equation}
\bm{p^{(t)}} = [\bm{I}_M-\bm{N}^T(\bm{NN}^T)^{-1}\bm{N}]\bm{g^{(t)}}, \label{dic}
\end{equation}
where $\bm{g^{(t)}}$ is the $M \times 1$ gradient vector, and the $m$-th element of $\bm{g^{(t)}}$ is $g(c_m^{(t)})$ given in (\ref{gcm}).
The step length $s^{(t)}$ in the direction $\bm{p^{(t)}}$ is determined by the following searching problem.
\begin{subequations}\label{stp}
	\begin{align}
	\mathop{\max }_{s^{(t)}}& \quad  \mathcal{F}(\bm{c} +s^{(t)}\bm{p^{(t)}} |\bar{x},\bar{y}) \\
	&\quad  \bm{0} \leq \bm{c} +s^{(t)}\bm{p^{(t)}} \leq \bm{\lambda}, \\
	&\quad  s^{(t)} \geq 0,
	\end{align}
\end{subequations}
where $\bm{\lambda} = [\lambda_1, \cdots, \lambda_M]$, and the optimal value of $s^{(t)}$ can be found by the bisection method since the objective function is increasing with respect to $s^{(t)}$.

{For any given $(\bar{x},\bar{y})$, the optimal caching density and a corresponding optimal objective can be obtained by solving Problem $\mathcal{P}3$.
Then, the optimal $(\bar{x}^*,\bar{y}^*)$ is given by
\begin{equation}
(x^*,y^*) = \arg \max_{\{\bar{x}, \bar{y}\}} \mathcal{F}^*(\bm{c}^*|\bar{x},\bar{y})
\end{equation}
However, the optimal value of $(x^*,y^*)$ cannot be analytically derived, so that the two-dimensional search is conducted.
With a small enough searching step size, the global optimal solution $(x^*,y^*)$ can be obtained.}
To reduce searching complexity, for given $x=\bar{x}$, the searching interval $[y_{min},y_{max}]$ of $y$ can be determined as
\begin{eqnarray}
y_{d}(\bar{x})&=&\min\{v_m\}\bar{x}, \\
y_{u}(\bar{x})&=&\min\left\{\max\{v_m\}\bar{x},\varLambda \right\},
\end{eqnarray}
where $\varLambda=\sum_{m=1}^M v_m \lambda_m$.

\begin{algorithm}
	\caption{ Algorithm to obtain the global optimum to Problem $\mathcal{P}1$}
	\begin{algorithmic}[1]\label{alg0}
		\REQUIRE  Searching stepsizes $\delta_x$ and $\delta_y$, and convergence condition $\delta_c$.
		\FOR{$\bar{x} = 0, \delta_x, 2\delta_x, \cdots, \left \lfloor \frac{\lambda_0}{\delta_x}\right \rfloor\delta_x$}
		\FOR{$\bar{y} = y_{d}(\bar{x}), \delta_y, 2\delta_y, \cdots, \left \lfloor \frac{y_{u}(\bar{x})-y_{d}(\bar{x})}{\delta_y} \right \rfloor\delta_y$}
		\STATE Initialize $t=0$, and $\bm{c}^{(0)}$ satisfying constraints (\ref{p01st1}) and (\ref{p01st2});
		\REPEAT
		\STATE Update the direction $\bm{p^{(t)}}$ and step length $s^{(t)}$ according to (\ref{dic}) and (\ref{stp}), respectively;
		\STATE Update $\bm{c}^{(t+1)} = \bm{c}^{(t)} + s^{(t)}\bm{p^{(t)}}$;
		\STATE Update $\mathcal{F}(\bm{c}^{(t+1)}|\bar{x},\bar{y})$;
		\STATE $t = t +1$;	
		\UNTIL  $\mathcal{F}(\bm{c}^{(t)}|\bar{x},\bar{y}) - \mathcal{F}(\bm{c}^{(t-1)}|\bar{x},\bar{y}) \leq \delta_c$;
		\STATE Update $U(\bar{x},\bar{y})=\left(\lambda_0-\bar{x}\right) \mathcal{F}(\bm{c}^{(t)}|\bar{x},\bar{y})$;
		\ENDFOR
		\ENDFOR
		\ENSURE   $\max_{(\bar{x},\bar{y}) \in \mathbb{Z}} U(\bar{x},\bar{y})$
	\end{algorithmic}
\end{algorithm}

The above analysis is summarized in the following Algorithm \ref{alg0}.

{The computational complexity of Algorithm \ref{alg0} mainly involves two parts:
 the gradient projection method and the two-dimensional search for optimal $(x^*,y^*)$.
For one execution of the gradient projection method, the complexity is dominated by the update of direction $\bm{p^{(t)}}$ in (\ref{dic}) in step 5, which needs about
$\frac{2}{3}2^3+ 8M2^2+8M2^3+2M^2=\mathcal{O}(3M^2+192M)$ flops (floating-point operations).
In addition, according to \cite{bertsekas1999nonlinear}, to get within  an $\varepsilon$-neighborhood of the optimal objective value, the iteration number of the gradient projection method is $\mathcal{O}(1/\varepsilon)$.
Meanwhile, for the two-dimensional search,	the iteration number is $\mathcal{O}\left(\left \lfloor \frac{\lambda_0}{\delta_x}\right \rfloor \left \lfloor \frac{\varLambda}{\delta_y}\right \rfloor\right)$ with step sizes $\delta_x$ and $\delta_y$ for $x$ and $y$, respectively.
Then, the total complexity of Algorithm \ref{alg0} is
$\mathcal{O}\left((3M^2+192M)\frac{1}{\varepsilon} \left \lfloor \frac{\lambda_0}{\delta_x}\right \rfloor \left \lfloor \frac{\varLambda}{\delta_y}\right \rfloor\right)$.}

Although the global optimum can be found by Algorithm \ref{alg0}, note that Algorithm \ref{alg0} has a high computational complexity.
Therefore, in the following, we consider an asymptotic approximation to obtain a low-complexity solution.

\subsection{Asymptotic Case}
Since the effect of transmission distance dependent path-loss has already been included in the D2D success probability, we consider the asymptotic D2D success probability with the relaxation of D2D cooperative distance limitation, i.e., $R \to \infty$.
The simulation results have verified this approximation by checking the gaps between the derived D2D success probability and the asymptotic case.

The asymptotic D2D success probability denoted as $\mathbb{P}_{s}^{\infty}$ is given by
\begin{equation}
\mathbb{P}_{s}^{\infty} = \sum_{m=1}^M \frac{c_mv_m}{\sum_{i=1}^{M} c_i v_i + \lambda_B \theta_Bv_m +c_m \rho_m^{\infty}\theta_I v_m}. \label{PssInfty}
\end{equation}

Similarly, we define $x = \sum_{m =1}^M c_m$.
However, in this case, for given $x = \bar{x}$, the active ratio $\rho_m^{\infty}$ is given by
\begin{equation}
\rho_m^{\infty} = 1-\left(1+\frac{v_m(\lambda_0-\bar{x})}{3.5\sum_{i=1}^M{v_ic_i}} \right)^{-3.5}, \quad m \in \mathcal{M}. \label{rho}
\end{equation}
Considering this sophisticated expression given in (\ref{rho}) makes $\mathbb{P}_{s}^{\infty}$ intractable, we introduce the following upper bound for the active ratio $\rho_m$ in each group as
\begin{equation}
\overline{\rho}_m= 1-\left(1+\frac{v_m(\lambda_0-x)}{3.5y^{\min}}  \right)^{-3.5},
\end{equation}
where $y^{\min}$ is the optimal objective value of the following Problem (\ref{hFunPro}).
\begin{subequations}\label{hFunPro}
	\begin{align}
	\mathop{\min }_{\bm{c}}& \quad y=\sum_{i=1}^M{v_ic_i} \\
	\textrm{s.t.} &\quad \sum_{m =1}^M c_m = \bar{x},  \\
	&\quad 0 \leq c_m \leq \lambda_m , \textrm{ for all } m \in  \mathcal{M}.
	\end{align}
\end{subequations}
It is easy to see that Problem (\ref{hFunPro}) is a linear programming, which can be solved by the standard linear programming method, such as the simplex method \cite{boyd2004convex}.
Since $\mathbb{P}_{s}^{\infty}$ decreases as $\rho_m^{\infty}$ increases, substituting $\overline{\rho}_m$ into (\ref{PssInfty}) yields a lower bound of $\mathbb{P}_{s}^{\infty}$ as
\begin{equation}
\underline{\mathbb{P}^{\infty}_{s}} = \sum_{m=1}^M \frac{c_mv_m}{\sum_{i=1}^{M} c_i v_i + \lambda_B \theta_Bv_m +c_m\overline{\rho}_m\theta_I v_m}.
\end{equation}
Consequently, Problem $\mathcal{P}1$ in this case is simplified into
\begin{subequations}\label{P31}
	\begin{align}
	\mathcal{P}4:\mathop{\max }_{\bm{c}}
	& \quad \mathcal{U}^{\infty} = \left(\lambda_0-\sum_{i=1}^Mc_i\right) \underline{\mathbb{P}^{\infty}_{s}} \\
	\textrm{s.t.} &\quad 0\leq c_m \leq \lambda_m , \textrm{ for all } m \in  \mathcal{M}.
	\end{align}
\end{subequations}
The performance gap between $\mathbb{P}_{s}$ in (\ref{Eq_Pss}) and $\mathbb{P}_{s}^{\infty}$, and the gap between the optimal offloading gain of Problem $\mathcal{P}1$ and the optimal $\mathcal{U}^{\infty}$ of Problem $\mathcal{P}4$ will be studied by simulations.

By fixing $x= \bar{x}$, Problem $\mathcal{P}4$ is simplified to the following Problem $\mathcal{P}4.1$ and the optimal value of $x$ can be found by one-dimension search over interval $[0,\lambda_0]$.
\begin{subequations}
	\begin{align}
	\mathcal{P}4.1:\mathop{\max }_{\bm{c}}
	& \quad   \sum_{m=1}^M \frac{c_mv_m}{\phi_m(\bm{c})} \label{p31Obj} \\
	\textrm{s.t.} &\quad \bar{x}= \sum_{i=1}^M c_i ,  \label{st31} \\
	&\quad 0 \leq c_m \leq \lambda_m , \textrm{ for all } m \in  \mathcal{M},
	\end{align}
\end{subequations}
where $\phi_m(\bm{c})= \sum_{i=1}^{M} c_i v_i + \lambda_B \theta_Bv_m +\overline{\rho}_mc_m \theta_I v_m$.

The objective of $\mathcal{P}4.1$ is to maximize the summation of fractional functions, so that Problem $\mathcal{P}4.1$ is the non-convex sum-of-ratios optimization \cite{7983007,6680785}.
Dinkelbach method is used to solve the optimization problem with single fractional function, and it cannot be directly applied here\cite{7762065}.
We first transform the Problem $\mathcal{P}4.1$ into the following equivalent form as
\begin{subequations}\label{P3.2}
	\begin{align}
	\mathop{\max }_{\bm{c},\beta}
	& \quad \sum_{m=1}^M \beta_m  \\
	\textrm{s.t.} &\quad \frac{c_mv_m}{ \phi_m(\bm{c})} \geq \beta_m  \textrm{ for all } m \in  \mathcal{M}, \label{st_p4}\\
	&\quad \bar{x}= \sum_{i=1}^M c_i , \\
	&\quad 0 \leq c_m \leq \lambda_m , \textrm{ for all } m \in  \mathcal{M}.
	\end{align}
\end{subequations}

Then, we have the following proposition.
\begin{proposition}\label{pro_3}
	If $(\bm{c}^*, \bm{\beta}^*)$ is the solution to Problem $(\ref{P3.2})$ that satisfies the KKT conditions, then there exists $\bm{u}^*=[u^*_1,\cdots, u^*_M]$ such that $\bm{c}^*$ is the optimal solution to following problem for $\bm{\beta}=\bm{\beta}^*$ and $\bm{u}=\bm{u}^*$.
	\begin{subequations}\label{P4}
		\begin{align}
		\mathop{\max }_{\bm{c}}
		& \quad \mathcal{G}(\bm{c})= \sum_{m=1}^M u_m(c_mv_m- \beta_m \phi_m(\bm{c}))  \\
		\textrm{s.t.} &\quad \bar{x}= \sum_{i=1}^M c_i ,  \label{st_p4eq} \\
		&\quad 0 \leq c_m \leq \lambda_m , \textrm{ for all } m \in  \mathcal{M}, \label{stp4}
		\end{align}
	\end{subequations}
	where $\bm{\beta}=[ \beta_1, \cdots, \beta_M]$.
	$\bm{c}^*$ also satisfies the following equations for $\bm{\beta}=\bm{\beta}^*$ and $\bm{u}=\bm{u}^*$.
	\begin{equation}
	u_m = \frac{1}{\phi_m(\bm{c})}, c_mv_m- \beta_m \phi_m(\bm{c}) =0, \quad \forall m =1,\cdots,M. \label{Pro4Eq}
	\end{equation}
\end{proposition}
\begin{proof}
	Since $\phi_m(\bm{c})>0$, the constraint (\ref{st_p4}) is equivalent to
	\begin{equation}
	c_mv_m- \beta_m \phi_m(\bm{c}) \geq 0 \label{st_pp4}.
	\end{equation}
	The Lagrangian for Problem $(\ref{P3.2})$ is
	\begin{multline}
	\mathcal{L}(\bm{c}, \bm{\beta}, \bm{u}, \eta)
	= \sum_{m=1}^M \beta_m - \\
	\sum_{m=1}^M u_m (\beta_m \phi_m(\bm{c})-c_mv_m) -\zeta(\sum_{i=1}^M c_i -\bar{x}),
	\end{multline}
	where $\{u_m\}$ and $\zeta$ are the non-negative Lagrangian multiplexers associated with constraint (\ref{st_pp4}) and (\ref{st_p4eq}), respectively.
	
	Based on KKT optimality conditions, if $(\bm{c}^*, \bm{\beta}^* )$ is the optimal solution of Problem $(\ref{P3.2})$, then there exists $\bm{u}^*$ and $\zeta^*$ such that
	\begin{eqnarray}
	\!\!\!\!\!\!\!-\sum_{i =1}^M  \beta_i^* u_i^*v_i - u_m^*\beta_m^*v_m\theta_I+u_m^*v_m -\zeta^* =0,&\!\!\!\!\!\forall m \in \mathcal{M}   \label{KT1} \\
	\frac{\partial \mathcal{L}}{\partial \beta_m} =1-u_m^* \phi_m(\bm{c}^*) =0,&\!\!\!\!\!\forall m \in \mathcal{M}  \label{KT2}\\
	u_m^* (\beta_m^* \phi_m(\bm{c}^*)-c_m^*v_m)=0,&\!\!\!\!\!\forall m \in \mathcal{M} \label{KT3} \\
	\sum_{i=1}^M c_i^* -\bar{x}=0,&\label{KT4}  \\
	c_m^*v_m- \beta_m^* \phi_m(\bm{c}^*)\leq 0,&\!\!\!\!\!\forall m \in \mathcal{M} \label{KT5} \\
	0 \leq c_m^* \leq \lambda_m,&\!\!\!\!\!\forall m \in \mathcal{M}  \label{KT6} \\
	u_m^* \geq 0,&\!\!\!\!\!\forall m \in \mathcal{M}.  \label{KT7}
	\end{eqnarray}
	
	First of all, according to (\ref{KT2}), $u_m^* >0$ for all $m\in \mathcal{M}$ due to the fact that $\phi_m(\bm{c}) >0$  for all $m\in \mathcal{M}$.
	As a result, from (\ref{KT2}), (\ref{KT3}) and  (\ref{KT5}), we have
	\begin{equation}
	u_m^* = \frac{1}{\phi_m(\bm{c}^*)}, c_m^*v_m- \beta_m^* \phi_m(\bm{c}^*)=0.
	\end{equation}
	Equation (\ref{Pro4Eq}) in Proposition \ref{pro_3} is then validated.
	
	Furthermore, given $\bm{\beta}=\bm{\beta}^*$ and $\bm{u}=\bm{u}^*$, it can be seen that  (\ref{KT1}), (\ref{KT4}), (\ref{KT7}) and (\ref{KT7}) are the KKT conditions for Problem (\ref{P4}).
	For $u_m >0,\beta_m >0, \forall m\in \mathcal{M}$, Problem (\ref{P4}) is a linear programming, which is convex.
	Consequently, the KKT conditions are also sufficient conditions of the optimal solution.
	In other words, $\bm{c}^*$ is the optimal solution to Problem (\ref{P4}) with $\bm{\beta}=\bm{\beta}^*$ and $\bm{u}=\bm{u}^*$.
\end{proof}

\begin{algorithm}
	\caption{ Iterative Algorithm to solve Problem $\mathcal{P}4.1$}
	\begin{algorithmic}[1]\label{alg1}
		\STATE Initialize $\bm{c}^{(0)}$ such that it satisfies the constraints (\ref{st_p4eq}), (\ref{stp4}).
		\STATE Set $t=0$, $\zeta \in (0,1)$ and $\epsilon \in (0,1)$.
		\STATE Initialize $\bm{\beta}^{(0)}, \bm{u}^{(0)}$ as
		\begin{align}
			&u_m^{(0)}=\frac{1}{\phi_m(\bm{c}^{(0)})},\\
			&\beta_m^{(0)}=\frac{c_m^{(0)} v_m}{\phi_m(\bm{c}^{(0)})}, \quad \forall m\in \mathcal{M};
		\end{align}
		
		\REPEAT
		\STATE Update ${\bm{c}^*}^{(t+1)}= \arg \underset{\bm{c}}{\max} \mathcal{G}(\bm{c}|\bm{u}^{(t)},\bm{\beta}^{(t)})$;
		\STATE Update $u_m^{(t+1)}$ and $\beta_m^{(t+1)}$ as
		\begin{eqnarray}
		&&\!\!\!\!\!\!\!\!\!u_m^{(t+1)} =  u_m^{(t)} - \zeta^{i^{(t+1)}} \frac{\chi_m(u_m^{(t)})}{\phi_m({\bm{c}^*}^{(t+1)})},\forall m \in \mathcal{M}, \\
		&&\!\!\!\!\!\!\!\!\!\beta_m^{(t+1)}= \beta_m^{(t)} - \zeta^{i^{(t+1)}} \frac{\kappa_m(\beta_m^{(t)})}{\phi_m({\bm{c}^*}^{(t+1)})},\forall m \in \mathcal{M};
		\end{eqnarray}
		where $i^{(t+1)}$ is the smallest integer among $i \in \{1,2,3,\cdots \}$satisfying
		\begin{align}
		&\sum_{m =1}^M \left|\chi_m\left(u_m^{(t)}-\zeta^{i} \frac{\chi_m(u_m^{(t)})}{\phi_m({\bm{c}^*}^{(t+1)})}\right)\right|^2
		\nonumber \\
		&+\sum_{m =1}^M \left|\kappa_m\left(\beta_m^{(t)}-\zeta^{i}
		\frac{\kappa_m(\beta_m^{(t)})}{\phi_m({\bm{c}^*}^{(t+1)})}\right)\right|^2 \nonumber \\
		&\leq \left(1-\epsilon\zeta^{i}\right)^2\left( \sum_{m =1}^M \left|\chi_m(u_m^{(t)})\right|^2 + \sum_{m =1}^M \left|\kappa_m(\beta_m^{(t)})\right|^2\right) ;
		\end{align}
		\STATE  Update $t = t +1$;
		\UNTIL  the following condition are satisfied.
		\begin{eqnarray}
		&&-1 + u_m^{(t)}\phi_m(\bm{c}^{(t)}) =0,\\
		&&\beta_m^{(t)} \phi_m(\bm{c}^{(t)}) -c_m^{(t)}v_m = 0,\forall m \in \mathcal{M};
		\end{eqnarray}
		\ENSURE ${\bm{c}^*}^{(t)}$, $\bm{\beta}^{(t)}$, and $\bm{u}^{(t)}$.
	\end{algorithmic}
\end{algorithm}

According to Proposition \ref{pro_3}, Problem (\ref{P3.2}) can be converted to Problem $(\ref{P4})$ with parameters $\bm{u}$ and $\bm{\beta}$.
Problem (\ref{P4}) is a linear programming with respect to $\bm{c}$, which can be solved by standard linear programming methods.
Therefore, we propose the following alternative algorithm to solve Problem $(\ref{P3.2})$.
Specifically, Algorithm \ref{alg1} consists of two key steps:
The first step is to obtain $\bm{c}^* = \arg \underset{\bm{c}}{\max} \mathcal{G}(\bm{c})$  for given $\bm{u}$ and $\beta$.
The second step is to update the parameters $\bm{u}$ and $\bm{\beta}$ according to the modified Newton method until the convergence conditions are satisfied.
Since Problem $\mathcal{P}4.1$ is the sum-of-ratios optimization, the proposed Algorithm \ref{alg1} converges to the solution satisfying KKT conditions as proved in \cite{jong2012efficient}.	
The detailed steps are summarized in Algorithm \ref{alg1}, where
\begin{equation}
\chi_m(u_m) = -1 + u_m\phi_m(\bm{c}), \kappa_m(\beta_m) = \beta_m\phi_m(\bm{c}) -c_mv_m.
\end{equation}
In summary, with given $\bar{x}$, the caching strategy $\bm{c}^*$ can be obtained by Algorithm \ref{alg1}.
Then the one-dimension search is adopted to find the optimal $x^*$ in the interval $[0,\lambda_0]$.
Consequently, a asymptotic caching strategy to Problem $\mathcal{P}1$ can be obtained.

{ The complexity of Algorithm \ref{alg1} is mainly dependent on step 5, and all the other steps provide explicit expressions.
For simplicity, we also adopt gradient projection method to solve the linear programming Problem (\ref{P4}), and the iteration number required for step 5 to obtain an $\varepsilon$-neighborhood of the optimal objective value is $\mathcal{O}(1/\varepsilon)$.}

\section{Numerical Results}

In this section, we first validate the analytical results derived in Section III by simulations.
Then we study the performance of Algorithm 2 and the impacts of different system parameters on the offloading gain and caching strategy.
Unless specified otherwise, the following simulation parameters are used.
The transmit power of users and BSs are set to 15 dBm and 20 dBm, respectively.
The density of BSs is set to $\lambda_B =10^{-4}$ per m$^2$.
{The offloading gain is evaluated over a $100\times 100$ m$^2$ area.}

\subsection{D2D Success Probability}

In this subsection, we validate the accuracy of the derived D2D success probability and its asymptotic approximation.
The impacts of the system parameters are also presented.

\begin{figure}
\centering
\includegraphics[width=1\linewidth]{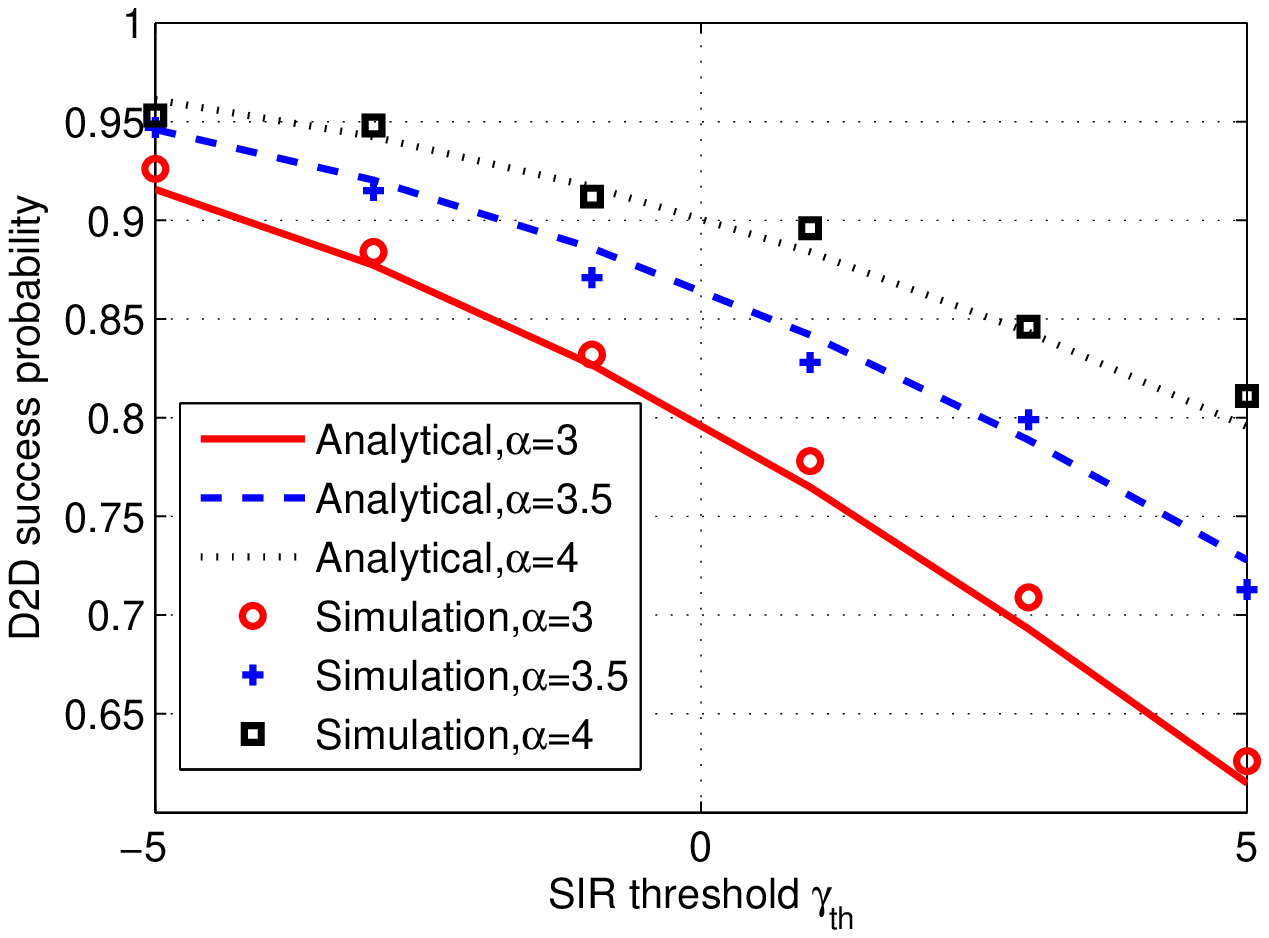}
\caption{D2D success probability performance.}
\label{fig1}
\end{figure}

\begin{figure}
	\centering
	\includegraphics[width=1\linewidth]{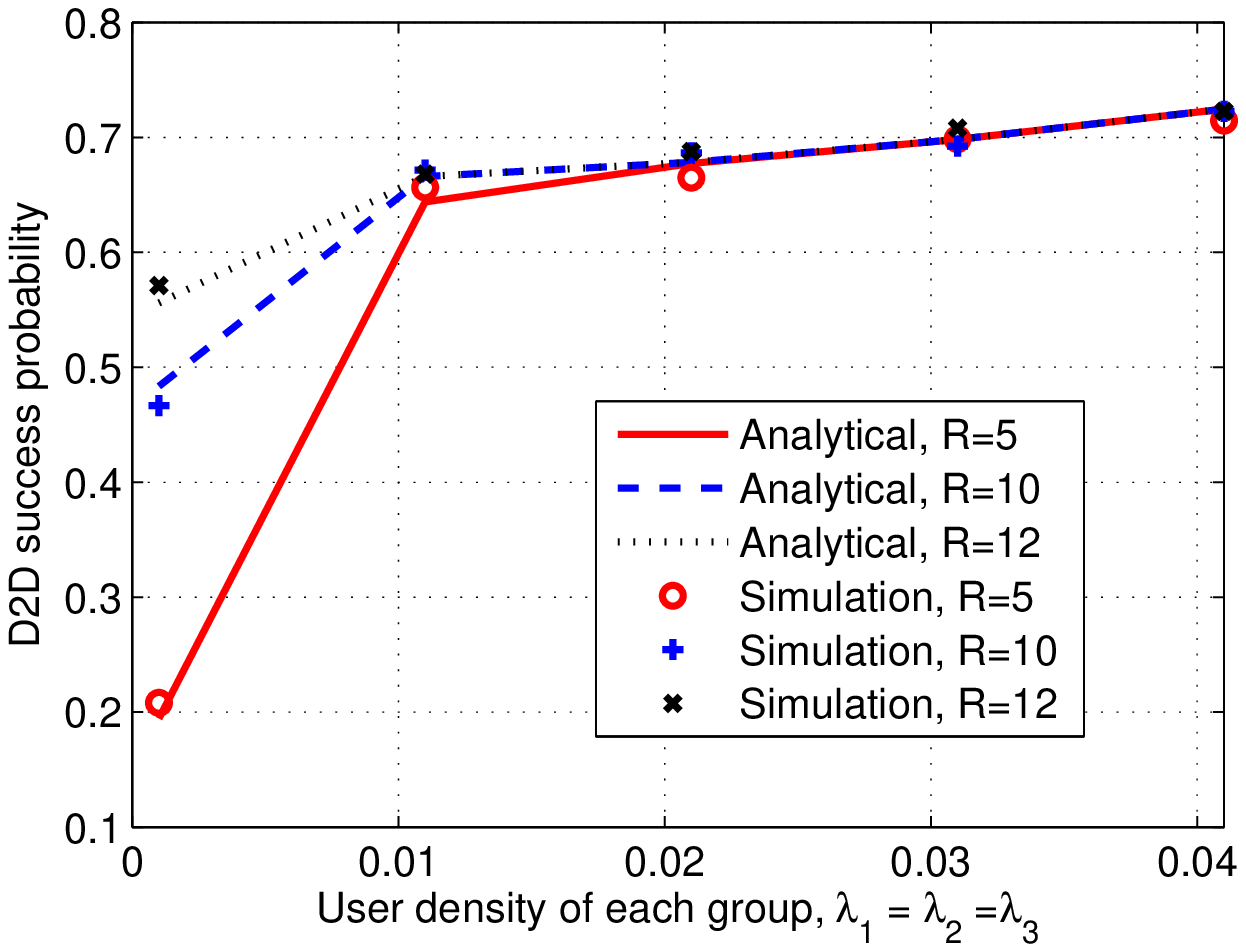}
	\caption{D2D success probability performance.}
	\label{fig2}
\end{figure}

Fig. \ref{fig1} compares the analytical expression of D2D success probability $\mathbb{P}_s$ in (\ref{Eq_Pss}) with the simulated ones, which are averaged over 2000 random realizations.
In Fig. \ref{fig1}, $M=3$, $\lambda_m=0.1, \forall m \in \mathcal{M}$, $B_1=0.1,B_2=0.3,B_3=0.6$, and $\bm{c}=[0.05, 0.09, 0.08]$.
It is shown that the analytical result $\mathbb{P}_{s}$ is very close to the simulated ones for all considered cases, which validates the accuracy of Proposition 3.
As expected, the D2D success probability decreases with the increasing SIR threshold and increases with the path-loss exponent $\alpha$.
The reason can be explained as follows.
Given the caching density, it is more difficult to satisfy the more strict SIR requirement when $\gamma_{th}$ increases, and the received interference at D2D receiver decreases when the path-loss exponent $\alpha$ increases.

Fig. \ref{fig2} illustrates the analytical results $\mathbb{P}_s$ and the simulated ones versus different user densities, where the user densities of different groups are set to be identical as shown in horizontal axis.
As expected, the analytical results are very close to the simulated results.
In addition, a large D2D transmission distance limit can notably improve the D2D success probability when the user density is small, while this benefit becomes inapparent with the large user density.
Moreover, although a high user density can obtain a high D2D success probability, the D2D success probability increases slowly for the high user density region.
This implies that a cost-effective caching strategy should be carefully designed to obtain a large offloading gain.

\begin{figure}

\includegraphics[width=1\linewidth]{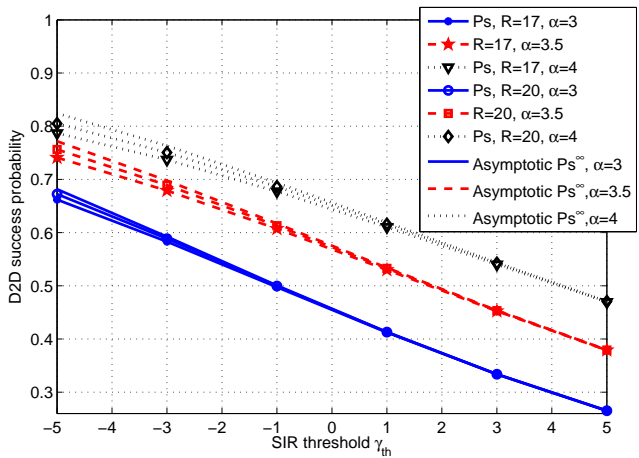}
\caption{D2D success probability performance versus the SIR threshold.}
\label{fig3}
\end{figure}

\begin{figure}

	\includegraphics[width=1\linewidth]{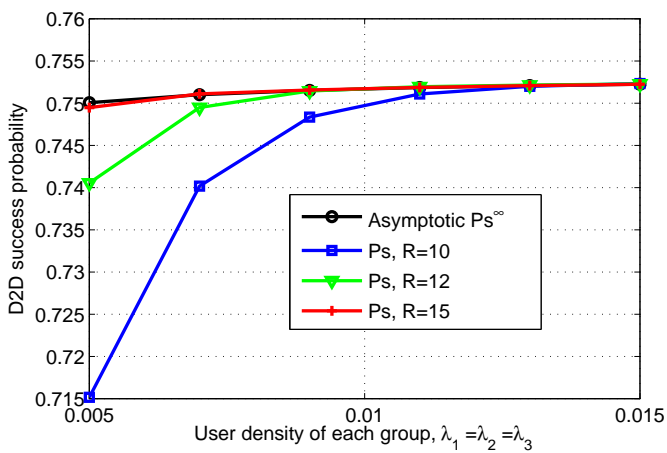}
	\caption{D2D success probability performance versus user density.}
	\label{fig4}
\end{figure}

Fig. \ref{fig3} illustrates the performance gap between D2D success probability $\mathbb{P}_s$ in (\ref{Eq_Pss}), and its asymptotic approximation $\mathbb{P}_{s}^{\infty}$ given in (\ref{PssInfty}).
The simulation parameters are set to be the same as that of Fig. \ref{fig1}.
As expected, the performance gap between $\mathbb{P}_{s}$ and $\mathbb{P}_{s}^{\infty}$ decreases with the maximal D2D transmission distance $R$, since the asymptotic approximation is obtained in the case of $R\to \infty$.
When the path-loss component $\alpha$ increases, the performance gap also decreases due to the fact that the impact of transmission distance decreases.
This validates the feasibility of adopting the asymptotic approximation $\mathbb{P}_{s}^{\infty}$ given in (\ref{PssInfty}).
In addition, it is shown that the performance gap between $\mathbb{P}_{s}$ and $\mathbb{P}_{s}^{\infty}$ also shrinks with the SIR threshold.
The reason is that the D2D success probability is mainly limited by the high SIR requirement when $\alpha$ is large.

Fig. \ref{fig4} illustrates the performance gap between $\mathbb{P}_s$ and $\mathbb{P}_{s}^{\infty}$ with respect to the user density, where the user densities of different groups are set to be identical as shown in horizontal axis.
The other simulation parameters are set to be the same as Fig. \ref{fig2}.
Obviously, the performance gap between $\mathbb{P}_s$ and $\mathbb{P}_{s}^{\infty}$ vanishes when the maximal D2D transmission distance $R$ increases.
Moreover, with given D2D distance, increasing the user density can also decrease the performance gap.
The reason is that the caching density also increases with the user density, since the caching density is set to $c_m=0.5\lambda_m$ for all $m \in \mathcal{M}$.
Consequently, the D2D success probability $\mathbb{P}_s$ keeps increasing with the caching density and the impact of D2D transmission distance on $\mathbb{P}_s$ is reduced.
As the caching density increases, $\mathbb{P}_s$ converges to $\mathbb{P}_{s}^{\infty}$, which also validates the accuracy of the asymptotic approximation.

\subsection{Offloading Gain}

In this subsection, we investigate the offloading gain achieved by proposed algorithms.
	
\begin{figure}
	\centering
	\includegraphics[width=1\linewidth]{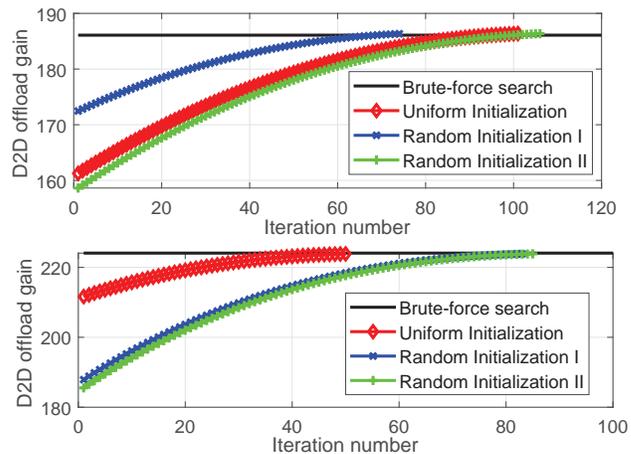}
	\caption{{Convergence behavior of Algorithm \ref{alg1}.}}
	\label{fig5}
\end{figure}

\begin{figure}
	\centering
	\includegraphics[width=1\linewidth]{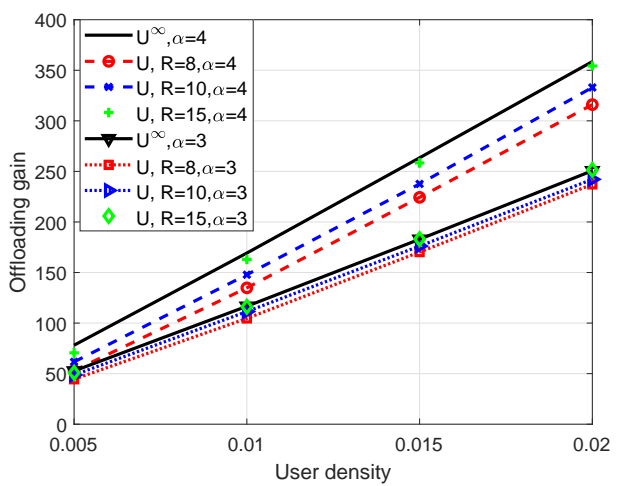}
	\caption{{The performance gap between optimal offloading gain $\mathcal{U}$ obtained by Algorithm \ref{alg0} and $\mathcal{U}^{\infty}$ obtained by Algorithm \ref{alg1}.}}
	\label{fig6}
\end{figure}

Fig. \ref{fig5} shows the convergence behaviors of Algorithm \ref{alg1} with { a given caching summation $\bar{x}=0.02$.
In Fig. \ref{fig5}, for $M=2$, the simulation parameters are set to be $\lambda_1=\lambda_2=0.02, B_1=0.1, B_2=0.9$, and for $M=3, \lambda_1=\lambda_2 =\lambda_3 = 0.02,B_1 = 0.1,B_2 = 0.4, B_3 = 0.5$.
The optimal solution to Problem $\mathcal{P}4$ for comparison is obtained by brute force search, the path-loss exponent is set to $\alpha = 3$, and the SIR requirement is $\gamma_{th}=3$ dB.}
As the formulated problem is nonconvex, different initializations may lead to different local maximums, so that the impacts of initialization methods are also presented.
In uniform initialization, the caching strategy is initialized by the optimal caching strategy achieved in unbiased case regardless of the trust bias.
In random initialization I and II, the initial caching densities are randomly generated with $\sum_{i=1}^{M}c_i =0.02$.
Numerical results illustrate that Algorithm \ref{alg1} always generates a nondecreasing sequence and converges to a fixed point.
Although different initializations slightly affect the performance and the convergence speed, it is shown that the proposed algorithm always converges to the optimal value within limited number of iterations for all considered cases.
In the following simulations, we adopt the uniform initialization method due to its stable convergence behavior.

Fig. \ref{fig6} elaborates the performance gap between the offloading asymptotic gain $\mathcal{U}^{\infty}$ obtained by Algorithm \ref{alg1} and the optimal offloading gain $\mathcal{U}$ obtained by Algorithm \ref{alg0}.
In Fig. \ref{fig6}, $M=2$, and the simulation parameters are set to be the same as that of Fig. \ref{fig5}.
As expected, the performance gap decreases with the maximal D2D transmission distance $R$.
It is observed that the offloading gain increases with user density as more users are interested in the reference content.
In addition, when path-loss component $\alpha$ increases, the offloading gain also increases due to a larger D2D success probability as shown in Fig. \ref{fig1}.
Moreover, the performance gap between $\mathcal{U}^{\infty}$ and $\mathcal{U}$ also increases as $\alpha$ increases.
The reason is that the performance gap is introduced by the impact of $R$, which can be magnified by a larger path-loss component $\alpha$.
Overall, the feasibility of adopting the asymptotic case to approximate the original case has been verified.
Then, in the following simulation, the maximal D2D distance is $R=15$m due to its tiny gap with the asymptotic approximation.

{For the D2D offloading networks, the problem of maximizing the density of successful receptions (i.e., offloading gain) is investigated in \cite{Malak} without considering the user preference and trust.
The heterogeneous user preference is incorporated in the caching strategy in \cite{Cooperative2017} but using the hit probability as an objective function.
However, for the D2D offloading system with the constraints of user preference and trust, as stated in Section I, the considered caching optimization problem to maximize the offloading gain has not been studied.
Therefore, to evaluate the performance of the proposed algorithm, the following caching policies are simulated for comparison:
\begin{itemize}	
\item ``OneUT'': This is a deterministic caching policy is given in \cite{Malak} for the one content case. Resulting from the fact that impacts of social characteristics are ignored, the approach in \cite{Malak} suggests that the density of "UT" approaches to zero and only one UT is enough for a user group. Then, the caching density of each group is set to be $c_m = \delta$, where $\delta$ is the searching stepsize of the optimal user density in Algorithm 1.
\item ``Uniform'': This is the caching policy developed in the unbiased case, and it can be viewed as an extension of the caching approach in \cite{Cooperative2017} considering multiple user groups with different interests, of which the approach is adjusted for a fair comparison by using the offloading gain as the objective function.
\end{itemize}}

\begin{figure}[htbp]
	\centering
	\includegraphics[width=1\linewidth]{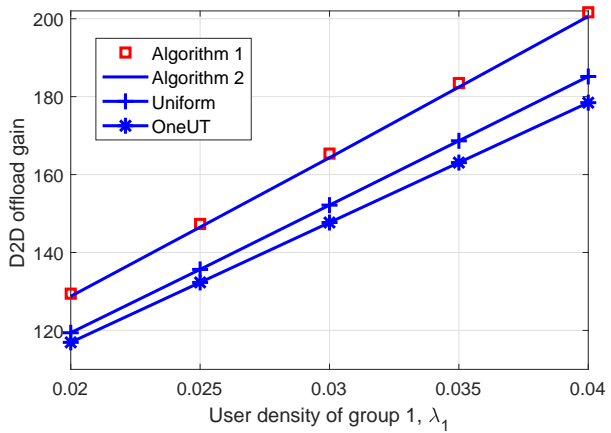}
	\caption{{The impact of user density on the offloading gain, M=2.}}
	\label{fig7}
\end{figure}

\begin{figure}[htbp]
	\centering
	\includegraphics[width=1\linewidth]{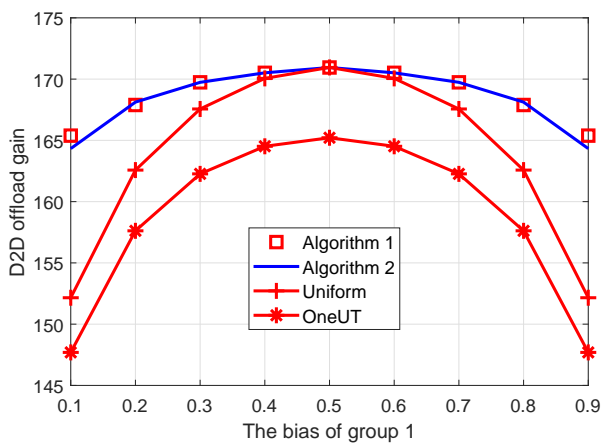}
	\caption{{The impact of trust bias on the offloading gain, M=2.}}
	\label{fig8}
\end{figure}

{
Fig. \ref{fig7} illustrates the impact of user density on the offloading gains achieved by different algorithms.
In Fig. \ref{fig7}, $M=2$, the user density of group 2 is 0.02 user per $m^2$, the maximal D2D distance is $R=15$m, and the other simulation parameters are same as those of Fig. \ref{fig6}.
The offloading gain increases with the user density $\lambda_1$, as increasing numbers of users are interested in the content.
It is observed that the asymptotic offloading gain achieved by Algorithm 2 well matches the optimal offloading gain obtained by Algorithm 1.
Also, the offloading gains achieved by the proposed algorithms outperform the other algorithms, resulting from the fact that social characteristics including the different interests and trust biases are considered in the proposed caching strategy.
Moreover, the offloading gain achieved by ``Uniform'' algorithm is larger than that of the ``OneUT'' algorithm, since the different user interests are considered in the ``Uniform'' algorithm.
It is worth pointing out that the performance gap increases with the user density $\lambda_1$.
This implies that the proposed algorithm can well adapt to the changes in system parameters and fully utilize the caching resources to achieve a larger offloading gain.}

{
In Fig. \ref{fig8}, the trust bias of group 1 denoted by $B_1$ is displayed on the x-axis, the trust bias of group 2 equals $1-B_1$, $\lambda_1 =0.04, \lambda_1 =0.02$, and the other simulation parameters are same as those of Fig. \ref{fig7}.
At the two ends of the x-axis, the difference between the trust bias of group 1 and the trust bias of group 2 is the largest.
It is observed that the offloading gain is notably affected by the different levels of trustworthiness among users, i.e., different bias values.
As the trustworthiness is neglected in the ``Uniform'' caching policy and ``OneUT'' caching policy,
the impact of trustworthiness on the obtained offloading gain is obvious.
Moreover, when the difference of the trustworthiness between user groups increases, the offloading gains achieved by these two algorithms degrade significantly.
The optimal offloading gain obtained by proposed Algorithm 1 matches the asymptotic offloading gain achieved by Algorithm 2, and they achieve larger offloading gains than the other two caching strategies.
However, the proposed algorithms can well adapt to the changes in trust bias, and the degradation of offloading gain is largely alleviated.
This is due to the fact that the social characteristics are fully exploited in the design of the proposed caching strategy.}

\begin{figure}
\centering
\includegraphics[width=0.9\linewidth]{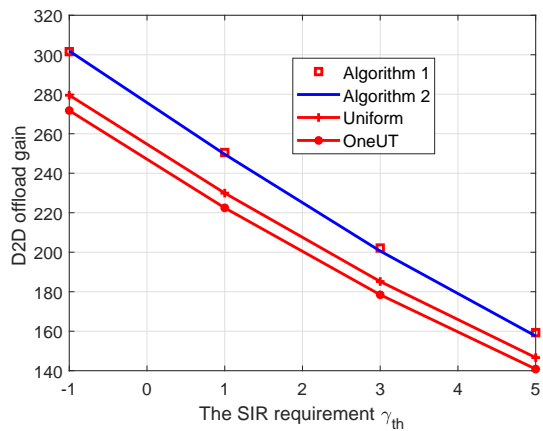}
\caption{{The impact of SIR threshold on the offloading gain, M=2.}}
\label{fig9}
\end{figure}

Fig. \ref{fig9} shows the impact of SIR threshold on the offloading performance.
In Fig. \ref{fig9}, the simulation parameters are set to be $\lambda_1=0.03,\lambda_2=0.01,B_1 =0.1,B_2 =0.9, \alpha = 3$.
It can be seen that the D2D offloading gain decreases when the SIR threshold $\gamma_{th}$ increases, due to the decreasing D2D success probability.
As expected, the proposed algorithms outperforms the other two caching strategies for all the considered configurations.
Moreover, we can see that the performance gap decreases as the SIR threshold $\gamma_{th}$ increases.
This is because as $\gamma_{th}$ increases, the D2D success probability decreases as shown in Fig. \ref{fig1}.

\begin{figure}
	\centering
	\includegraphics[width=0.9\linewidth]{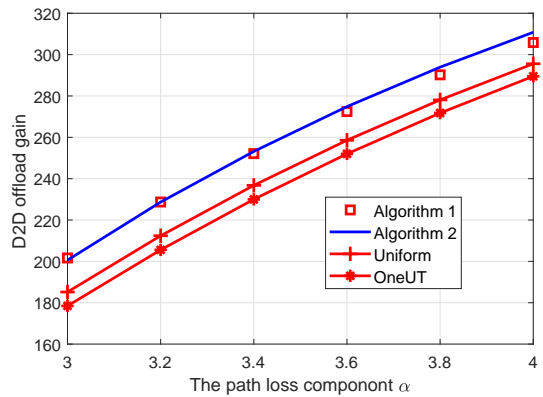}
	\caption{{The impact of path-loss component on the offloading gain, M=2.}}
	\label{fig10}
\end{figure}

Fig. \ref{fig10} shows the impact of path-loss component on the offloading performance.
In Fig. \ref{fig10}, $\gamma_{th} = 3$dB, and the other simulation parameters are same as Fig. \ref{fig9}.
When the path-loss component $\alpha$ increases, the offloading gain increases due to a increasing D2D success probability as shown in Fig. \ref{fig1}.
It is observed that the performance gaps between the proposed Algorithm
Moreover, it is interesting to see that the performance gap keeps invariant in Fig. \ref{fig10}, when the path-loss component changes.
This implies that the performance gap is mainly determined by the social characteristics and the SIR threshold $\gamma_{th}$.

\subsection{Caching Strategy}

\begin{figure}
\centering
\includegraphics[width=1\linewidth]{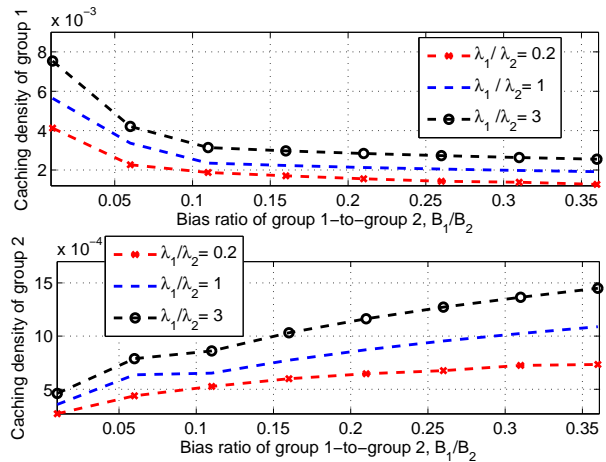}
\caption{The impact of user density and trust bias on the caching density, where $M=2$, $\alpha = 3, \lambda_2=0.1$ and $\gamma_{th}=3$ dB.}
\label{fig11}
\end{figure}

We investigate the impacts of user density and trust bias on the caching density achieved by Algorithm \ref{alg1} in Fig. \ref{fig11}.
For the impact of user density, it is observed that the caching probability of group 1 increases with the density ratio $\lambda_1/\lambda_2$, while the caching probability of group 2 decreases with the density ratio $\lambda_1/\lambda_2$.
This is due to that the relative density of interested user in group 1 increases and the the relative density of interested user in group 2 decreases, which shows a positive correlation between the caching density and user density.
For the impact of association bias, the caching probability of group 1 decreases with the bias ratio $B_1/B_2$, while the caching probability of group 2 increases the bias ratio $B_1/B_2$.
Moreover, the caching density of group 1 is larger than that of group 2 due to that $B_1 < B_2$, which implies that the caching density is inversely proportional to the bias value.
Since URs always prefer to associate with the group 2 due to its high association bias, a larger caching density of group 2 will introduce more severer interference than that of group 1.
Accordingly, more contents should be deployed to the groups with low trust bias values, to increase the D2D success probabilities contributed by these groups.

\begin{figure}
	\centering
	\includegraphics[width=1\linewidth]{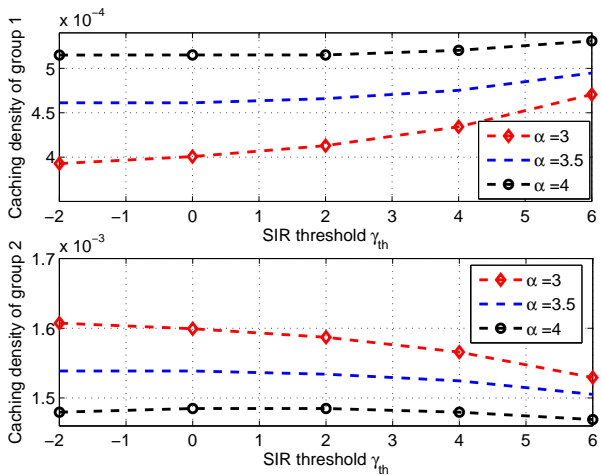}
	\caption{The impact of SIR threshold and path-loss component on the caching density, where $M=2$, $\lambda_1=\lambda_2=0.05$ and $B_1/B_2 = 10$.}
	\label{fig12}
\end{figure}

Fig \ref{fig12} shows the impacts of the SIR threshold $\gamma_{th}$ and path-loss component $\alpha$ on the caching density.
It is observed that caching density of group 1 increases as $\gamma_{th}$ and $\alpha$ increases, while the caching density of group 2 decreases as $\gamma_{th}$ and $\alpha$ increases.
The reason can be explained as follows.
Since $B_1>B_2$ in Fig \ref{fig12}, the caching density of group 2 is larger than that of group 1, which also leads to larger interference in group 2 than that in group 1.
Since a larger $\gamma_{th}$ requires a larger received SIR, the caching density of group 2 should decrease and that of group 1 should increase.
Besides, a larger $\alpha$ will reduce the impact of bias ratio $B_1/B_2$ according to Proposition 3, so that the difference of caching densities of differen groups reduces.

\section{Conclusion}
In this work, we considered the joint impact of social characteristics and physical transmission conditions on the D2D assisted content offloading performance.
The caching strategy was optimized to maximize the offloading gain, which was defined as the offloaded traffic via D2D communications.
The D2D success probability was first derived, of which the complicated expression made the formulated problem nonconvex and difficult to solve.
We first investigated the optimal caching strategy for the unbiased case, and then proposed a numerical algorithm to obtain the global optimal caching strategy for the general case.
To reduce complexity, a low-complexity iterative algorithm was then proposed to obtain the asymptotic caching strategy.
Finally, by simulations, the derived D2D successful transmission probability and its asymptotic approximation were validated.
Our proposed algorithm outperforms the existing caching strategies in terms of offloading gain, and it can well adapt to the changes of system parameters and wisely utilize the caching resources.
The obtained caching strategy is jointly determined by the social characteristics and physical transmission conditions.
A positive correlation is shown between the caching density and the user density, but the caching density is inversely proportional to the association bias.
{
When the dynamic nature of users is considered, D2D offloading networks are facing new challenges.
In addition, the offloading gain can be improved by optimizing D2D user association.
These issues will be investigated in our future work.}

\section*{Acknowledgment}
This work was supported in part by the National Science and Technology Major Project under Grant No. 2016ZX03001016-003, the National Natural Science Foundation of China under Grants No. 61871128, No. 61521061, No. 61571125, the Central University Basic Research fund under Grant No. 2242019K40202, the UK Royal Society Newton International Fellowship under Grant NIF$ \verb|\|$ R1$ \verb|\|$180777.

\bibliographystyle{IEEEtran}
\bibliography{IEEEabrv,Refer_06_19}

\end{document}